\documentclass[11pt]{article}
\usepackage{bcc24}
\usepackage{url}
\usepackage{graphicx}

\newcommand{\qitem}[1]{\noindent\leavevmode\hangindent10mm%
       \noindent\hbox to10mm{#1\hss}\ignorespaces}

\newenvironment{probl}[1]{\noindent\textsc{#1}\begin{list}{}{
      \setlength{\leftmargin}{20mm}\setlength{\labelwidth}{18mm}%
      \setlength{\itemsep}{2pt plus 1pt}\setlength{\parsep}{\parskip}%
      \setlength{\listparindent}{0pt}\setlength{\topsep}{0pt}}}%
  {\end{list}}

\makeatletter
\newcommand{\dege}{\mathop{\operator@font deg}}
\newcommand{\puz}{\mathop{\operator@font puz}}
\makeatother

\bcctitle{The Complexity of Change}
\bccname{Jan van den Heuvel}
\bccaddressa{Department of Mathematics\\
  London School of Economics\\
  London, UK}
\bccemaila{j.van-den-heuvel@lse.ac.uk}

\begin{document}

\makebcctitle

\begin{abstract}
Many combinatorial problems can be formulated as ``Can I transform
configuration~1 into configuration~2, if certain transformations only are
allowed?''. An example of such a question is: given two $k$-colourings of a
graph, can I transform the first $k$-colouring into the second one, by
recolouring one vertex at a time, and always maintaining a proper
$k$-colouring? Another example is: given two solutions of a SAT-instance,
can I transform the first solution into the second one, by changing the
truth value one variable at a time, and always maintaining a solution of
the SAT-instance? Other examples can be found in many classical puzzles,
such as the 15-Puzzle and Rubik's Cube.

In this survey we shall give an overview of some older and more recent work
on this type of problem. The emphasis will be on the computational
complexity of the problems: how hard is it to decide if a certain
transformation is possible or not?
\end{abstract}

\section{Introduction}\label{sec-intro}

\emph{Reconfiguration problems} are combinatorial problems in which we are
given a collection of configurations, together with some transformation
rule(s) that allows us to change one configuration to another. A classic
example is the so-called \mbox{\emph{15-puzzle}} (see
Figure~\ref{fig:15puzzle}): 15~tiles are arranged on a $4\times4$ grid,
with one empty square; neighbouring tiles can be moved to the empty slot.
The normal aim is, given an initial configuration, to move the tiles to the
position with all numbers in order (right-hand picture in
Figure~\ref{fig:15puzzle}). Readers of a certain age may remember Rubik's
cube and its relatives as examples of reconfiguration puzzles (see
Figure~\ref{fig:Rubik}).

\begin{figure}[ht]
\medskip
  \centering
  \includegraphics[width=40mm]{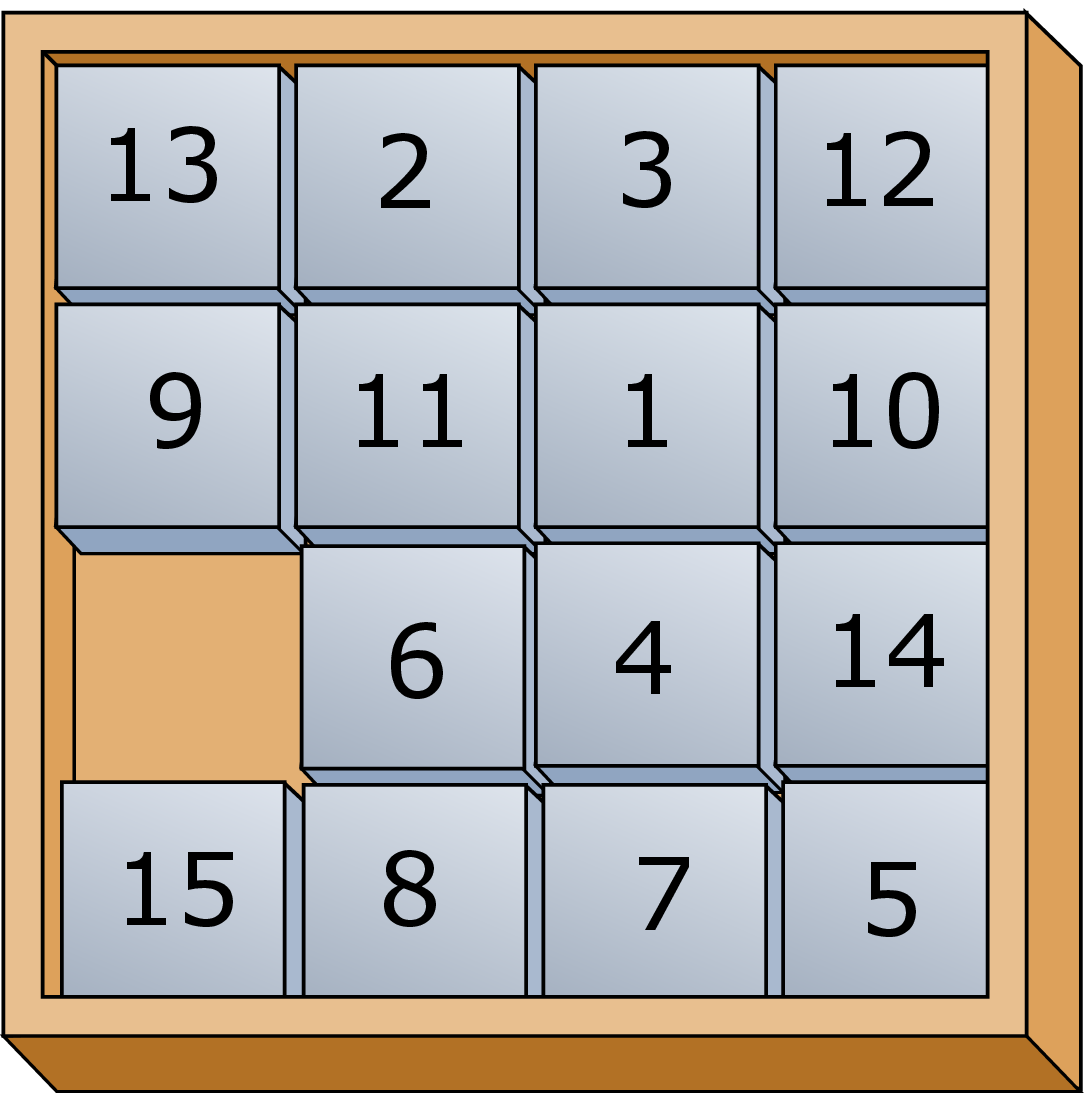}
  \hspace{15mm}
  \includegraphics[width=40mm]{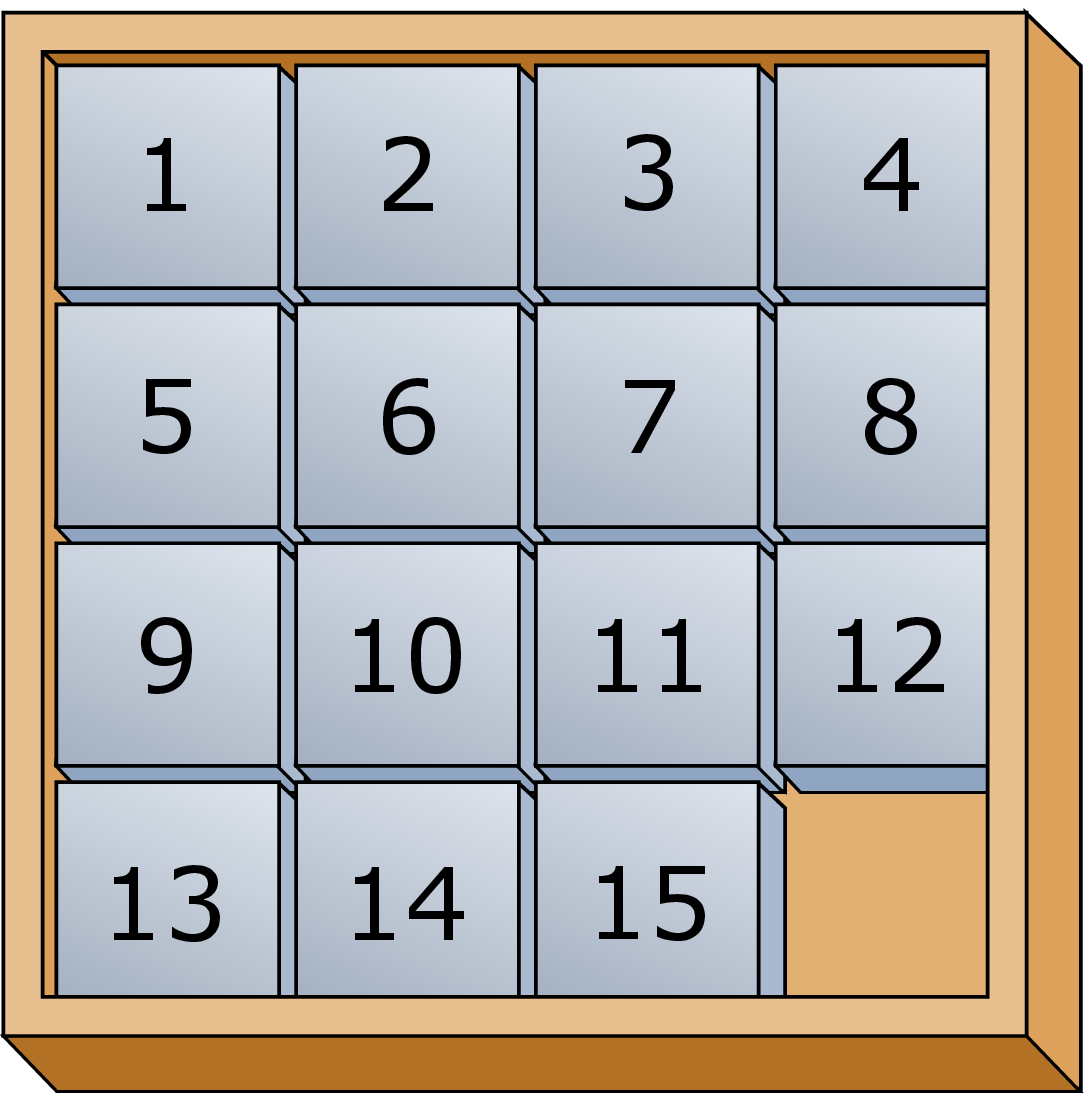}
  \medskip
  \caption[Two configurations of the 15-puzzle]{Two configurations of the
    15-puzzle (left picture \copyright\ 2008 Theon, right picture
    \copyright\ 2006 Booyabazooka; via Wikipedia)}
  \label{fig:15puzzle}
\end{figure}

\begin{figure}[ht]
\medskip
  \centering
  \includegraphics[width=40mm]{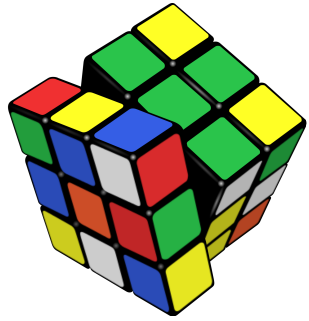}
  \medskip
  \caption{Rubik's cube (\copyright\ 2006 Booyabazooka; via Wikipedia)}
  \label{fig:Rubik}
\end{figure}

More abstract kinds of reconfiguration problems abound in graph theory. For
instance, suppose we are given a planar graph and two 4-colourings of that
graph. Is it possible to transform the first 4-colouring into the second
one, by recolouring one vertex at a time, and never using more than 4
colours? Taking any two different 4-colourings of the complete graph~$K_4$
shows that the answer is not always yes. But what would happen if we
allowed a fifth colour? And whereas it is easy to see what the situation is
with two 4-colourings of~$K_4$, how hard is it to decide in general if two
given 4-colourings of some planar graph can be transformed from one to the
other by recolouring one vertex at a time?

As a final (class of) example in this introduction we mention
reconfiguration problems on satisfiability problems. Given some Boolean
formula and two satisfying assignments of its variables, is it possible to
transform one assignment into the other by changing the value of one
variable at a time, but so that the formula remains TRUE during the whole
sequence of transformations?

\medskip
In this survey we concentrate on \emph{complexity} considerations of
transformation problems. In other words, we are interested in knowing
\emph{how hard} it is computationally to decide if the answer to some
problem involving transformation is yes or no. More specifically, we will
look at two types of those complexity question, which very roughly can be
described as follows:

\bigskip
\begin{probl}{A-to-B-Path}
\item[\normalfont\textit{Instance}:\hfill] Description of a collection of
  feasible configurations;\\
  description of one or more transformations changing one configuration to
  another;\\
  description of two feasible configurations $A,B$.
\item[\normalfont\textit{Question}:\hfill] Is it possible to change
  configuration~$A$ into configuration~$B$ by a sequence of transformations
  in which each intermediate configuration is a feasible configuration as
  well?
\end{probl}

\bigskip
\begin{probl}{Path-between-All-Pairs}
\item[\normalfont\textit{Instance}:\hfill] Description of a collection of
  feasible configurations;\\
  description of one or more transformations changing one configuration to
  another.
\item[\normalfont\textit{Question}:\hfill] Is it possible for any two
  feasible configurations $A,B$ to change configuration~$A$ into
  configuration~$B$ by a sequence of transformations in which each
  intermediate configuration is a feasible configuration as well?
\end{probl}

\bigskip\noindent
Of course, there are many other questions that can be asked: how many steps
does it take to go from one configuration to another? Which two
configurations are furthest apart? Etc., etc. Many of these questions have
been considered for particular problems, and where opportune we shall
mention some of this work.

An alternative way to formulate this type of problem is by using the
concept of a \emph{configuration graph}. This is the graph that has as
vertex set the collection of all possible feasible configurations, while
two configurations are connected by an edge if there is a transformation
changing one to the other. Note that nothing that we have said so far rules
out the possibility that the transformation goes one way only, but in
general we will assume that we can always go back and forth between
configurations. This means the configuration graph can be taken to be an
undirected graph.

Using the language of configuration graphs, the two general decision
problems above can be rephrased as follows.
\textsc{Path-between-All-Pairs}: is the configuration graph connected?
\textsc{A-to-B-Path}: given two vertices (configurations) in the
configuration graph, are they in the same component?

\medskip
In most of this survey we will use fairly informal language. So we may use
``step'' or ``move'' instead of ``transformation'' (a one-step change). On
the other hand, ``transform configuration~$A$ to configuration~$B$'', ``move
from~$A$ to~$B$'' or ``go from~$A$ to~$B$'' usually indicate a sequence of
transformations.

\subsection{A little bit on computational complexity}

This survey cannot give a full definition of the complexity classes we will
encounter, and we only give a general, intuitive, description of some of
them. The interested reader can find all details in appropriate textbooks,
such as Garey~\& Johnson~\cite{GJ} and Papadimitriou~\cite{Pa}.

We assume the reader is familiar with the concept of decision problems
(problems that have as answer either ``yes'' or ``no'') and the complexity
classes~\textbf{P}, \textbf{NP} and~\textbf{coNP}. We will also regularly
encounter the class \textbf{PSPACE}. A decision problem is in
\textbf{PSPACE}, or \emph{can be solved in polynomial space}, if there
exists an algorithm that solves the problem and that uses an amount of
memory that is polynomial in the size of the input. The related
non-deterministic complexity class \textbf{NPSPACE} is similarly defined as
the class of decision problems for which there exists a non-deterministic
algorithm that can recognise ``yes''-instances of the problem using an
amount of memory that is polynomial in the size of the input. For a
non-deterministic algorithm we mean by \emph{recognising ``yes''-instances}
that for every ``yes''-instance (but for none of the ``no''-instances)
there is a possible run of the algorithm that finishes in finite time with
a ``yes'' answer.

We obviously have
$\text{\textbf{P}}\subseteq\text{\textbf{NP}}\cap\text{\textbf{coNP}}$ and
$\text{\textbf{PSPACE}}\subseteq\text{\textbf{NPSPACE}}$, and a little bit
of thought should convince the reader that we also have
$\text{\textbf{NP}}\cup\text{\textbf{coNP}}\subseteq\text{\textbf{PSPACE}}$.
(Trial and error of all possible solutions of a problem in~\textbf{NP}
or~\textbf{coNP} can be done in polynomial space.) For most of these
inclusions it is unknown if they are proper inclusions or if the classes
are in fact the same, leading to some of the most important problems in
computer science (settling whether or not
$\text{\textbf{P}}=\text{\textbf{NP}}$ is worth a million
dollars~\cite{Cl}). The one exception is that we know that \textbf{PSPACE}
and \textbf{NPSPACE} are in fact the same by the celebrated theorem of
Savitch~\cite{Sa}.

Within each complexity class we can define so-called \emph{complete}
problems. Again, we refer to the appropriate textbooks for the precise
definition; for us it is enough intuitively to assume that these are the
``most difficult'' problems in their class.

\subsection{Computational complexity of reconfiguration problems}

In order to be able to ask sensible questions (and obtain sensible answers)
about the complexity of reconfiguration problems, we will make some
assumptions regarding their properties. In particular, when describing the
possible configurations, we will assume that these are not given as a full
set of all configurations, but by some compact description. Otherwise, if
the set of all possible configurations was part of the input, most decision
problems about those configurations would trivially be possible in
polynomial time because the input would be very large.

More precisely, we assume that an instance of the input contains an
algorithm to decide if a candidate configuration really is feasible or not.
Similarly, we are in general not interested in problems where the
collection of possible transformations needs to be given (in the form of a
list of all pairs that are related by the transformation). Instead, we
assume that the input contains an algorithm to decide, given two
configurations, whether or not we can get the second configuration from the
first by a single transformation.

Regarding these algorithmic issues of the description of an instance of a
configuration problem, we make the following assumptions:

\medskip
\begin{list}{}{
    \setlength{\leftmargin}{10mm}\setlength{\labelwidth}{8mm}%
    \setlength{\itemsep}{2pt plus 1pt}\setlength{\parsep}{\parskip}%
    \setlength{\listparindent}{0pt}\setlength{\topsep}{0pt}}
\item[\normalfont{A1}:\hfill] Deciding if a given possible configuration is
  a feasible configuration can be done in polynomial time.
\item[\normalfont{A2}:\hfill] Given two feasible configurations, deciding
  if there is a transformation from the first to the second can be done in
  polynomial time.
\end{list}

\medskip\noindent
Note that these assumptions guarantee that both of our general
reconfiguration problems are in \textbf{NPSPACE} (hence in
\textbf{PSPACE} by Savitch's Theorem). The following is a non-deterministic
algorithm for \textsc{A-to-B-Path} that would work in polynomial space,
when required to decide if it is possible to have a sequence of
transformations from configuration~$A$ to configuration~$B$:

\medskip
\begin{list}{}{
    \setlength{\leftmargin}{10mm}\setlength{\labelwidth}{8mm}%
    \setlength{\itemsep}{2pt plus 1pt}\setlength{\parsep}{\parskip}%
    \setlength{\listparindent}{0pt}\setlength{\topsep}{0pt}}
\item[\normalfont{1}:\hfill] Given the initial configuration~$A$,
  ``guess'' a next configuration~$A_1$. Check that~$A_1$ is indeed a
  feasible configuration and that there is transformation from~$A$
  to~$A_1$. If~$A_1$ is a valid next configuration, ``forget'' the initial
  configuration~$A$ and replace it by~$A_1$.
\item[\normalfont{2}:\hfill] Repeat the step above until the target
    configuration~$B$ has been reached.
\end{list}

\medskip\noindent
If there is indeed a way to go from configuration~$A$ to
configuration~$B$, then a sequence of correct guesses in the algorithm
above will indeed recognise that, using a polynomial amount of space; while
if there is no sequence of transformations from~$A$ to~$B$, then the
algorithm will never finish. To extend the algorithm to the
\textsc{Path-between-All-Pairs} problem, we just need to repeat this task
for all possible pairs. This means systematically generating all pairs of
candidate configurations, and testing those. Since each candidate
configuration has a size that is polynomial in the size of the original
input (as it can be tested in polynomial time whether or not a candidate
configuration is feasible), this brute-force generation of all possible
pairs of configurations and testing whether or not they are feasible and
connected can be done in a polynomial amount of memory as well.

The fact that all problems we consider are ``automatically'' in
\textbf{PSPACE} means that we are in particular interested in determining
if a particular variant is in a more restricted class (\textbf{P},
\textbf{NP}), or if it is in fact \textbf{PSPACE}-complete.

\medskip
A final property that all examples we look at will have is that the
transformations are symmetric: if we can transform one configuration into
another, then we can also go the other way round. There is no real reason
why this symmetry should always be the case, it is just that most
reconfiguration problems considered in the literature have this feature.
In particular, the sliding token problems we look at in
Section~\ref{sec-tok} could just as well be formulated for directed graphs,
leading in general to directed reconfiguration problems.

\section{Reconfiguration of satisfiability problems}\label{sec-sat}

A collection of interesting results regarding the reconfiguration of
solutions of a given Boolean formula were obtained by Gopalan et
al.~\cite{GK1,GK2}. They considered the following general set-up. Given a
Boolean formula~$\varphi$ with~$n$ Boolean variables, the feasible
configurations are those assignments from $\{T,F\}^n$ that
satisfy~$\varphi$ (i.e., for which~$\varphi$ gives the value TRUE); the
allowed transformation is changing the value of exactly one of the
variables.

The collection of all possible assignments, together with edges added
between those pairs that differ in exactly one variable, gives us the
structure of the (graph of the) $n$-dimensional hypercube. This means that
the configuration graph for a Boolean formula reconfiguration problem is
the induced subgraph of the hypercube induced by the satisfying
assignments. It is this additional structure that should give hope of a
better understanding of this type of reconfiguration problem.

The first to analyse the connectivity properties of the configuration
graphs of this type of problem were Gopalan et al.~\cite{GK1,GK2}. To
describe their results, we need a few more definitions.

A \emph{logical relation}~$R$ is a subset of $\{T,F\}^k$, where~$k$ is the
\emph{arity} of~$R$. For instance, if $R_{1/3}=\{TFF,FTF,FFT\}$, then
$R_{1/3}(x_1,x_2,x_3)$ is TRUE if and only if exactly one of $x_1,x_2,x_3$
is~$T$. For~$\mathcal{S}$ a finite set of logical relations, a
\emph{CNF($\mathcal{S}$)-formula} over a set of variables
$V=\{x_1,x_2,\ldots,x_n\}$ is a finite conjunction $C_1\wedge
C_2\wedge\dots\wedge C_m$ of clauses built using relations
from~$\mathcal{S}$, variables from~$V$, and the constants~$T$ and~$F$.
Hence each~$C_i$ is an expression of the form
$R(\xi_1,\xi_2,\ldots,\xi_k)$, where~$R$ is a relation of arity~$k$, and
each~$\xi_j$ is a variable from~$V$ or one of the Boolean constants $T,F$.
The satisfiability problem \textsc{Sat($\mathcal{S}$)} associated with a
finite set of logical relations~$\mathcal{S}$ asks: given a
CNF($\mathcal{S}$)-formula~$\varphi$, is it satisfiable?

As an example, if we use the relation~$R_{1/3}$ as above, and set
$\mathcal{S}_{1/3}=\{R_{1/3}\}$, then CNF($\mathcal{S}_{1/3}$) consists of
Boolean expressions of the form
\[\varphi=(x_i\vee x_j\vee x_k)\wedge(x_{i'}\vee x_{j'}\vee x_{k'})\wedge
\cdots.\]
Finally, such an expression~$\varphi$ is true for some assignment from
$\{T,F\}$ to the variables if and only if each clause $x_i\vee x_j\vee x_k$
has exactly one variable that is~$T$. The satisfiability problem
\textsc{Sat($\mathcal{S_{1/3}}$)} is known as \textsc{Positive-1-In-3-Sat};
a decision problem that is \textbf{NP}-complete~\cite{Sc}.

Another, better known, example is obtained by taking
$\mathcal{S}_2=\{\,\{TF,FT,TT\},$ $\{FF,FT,TT\},\,\{FF,TF,FT\}\,\}$. Here
the relation $R=\{TF,FT,TT\}$ indicates that $R(x_1,x_2)$ is TRUE when at
least one of $x_1,x_2$ is~$T$, hence it represents clauses of the form
$x_1\vee x_2$. Similarly, $\{FF,FT,TT\}$ represents clauses
$\neg x_1\vee x_2$, and $\{FF,TF,FT\}$ represents clauses
$\neg x_1\vee\neg x_2$. We see that CNF($\mathcal{S}_2$) is exactly the set
of Boolean expressions that can be formulated with clauses that are
disjunctions of two literals. (A \emph{literal} is one of~$x_1$
or~$\neg x_i$ for a variable $x_i\in V$.) We call such a formula a
\emph{2-CNF-formula}.

Similarly, the \emph{3-CNF-formulas} are exactly those formulas whose set
of relations is $\mathcal{S}_3=\{R_0,R_1,R_2,R_3\}$, where
\[\begin{array}{@{}c@{}}
  R_0\:=\:\{T,F\}^3\setminus\{FFF\},\quad
  R_1\:=\:\{T,F\}^3\setminus\{TFF\},\\
  R_2\:=\:\{T,F\}^3\setminus\{TTF\},\quad
  R_3\:=\:\{T,F\}^3\setminus\{TTT\}.\end{array}\]
Note that this also means that \textsc{Sat($\mathcal{S}_2$)} and
\textsc{Sat($\mathcal{S}_3$)} are equivalent to the well-known
\textsc{2-SAT} and \textsc{3-SAT} decision problems, respectively.

\medskip
Schaefer~\cite{Sc} proved a celebrated dichotomy theorem about the
complexity of \textsc{Sat($\mathcal{S}$)}: for certain sets~$\mathcal{S}$
(nowadays called \emph{Schaefer sets}), \textsc{Sat($\mathcal{S}$)} is
solvable in polynomial time; while for all other sets~$\mathcal{S}$ the
problem is \textbf{NP}-complete.

In~\cite{GK1,GK2}, the following two decision problems are considered for
given~$\mathcal{S}$.

\bigskip
\begin{probl}{st-Conn($\mathcal{S}$)}
\item[\normalfont\textit{Instance}:\hfill] A
  CNF($\mathcal{S}$)-formula~$\varphi$, and two satisfying
  assignments~$\textbf{s}$ and~$\textbf{t}$ of~$\varphi$.
\item[\normalfont\textit{Question}:\hfill] Is there a path
  between~$\textbf{s}$ and~$\textbf{t}$ in the configuration graph of
  solutions of~$\varphi$?
\end{probl}

\bigskip
\begin{probl}{Conn($\mathcal{S}$)}
\item[\normalfont\textit{Instance}:\hfill] A
  CNF($\mathcal{S}$)-formula~$\varphi$.
\item[\normalfont\textit{Question}:\hfill] Is the configuration graph of
  solutions of~$\varphi$ connected?
\end{probl}

\bigskip\noindent
The key concept for these problems appears to be that of a \emph{tight} set
of relations~$\mathcal{S}$ -- see~\cite{GK1,GK2} for a precise definition
of this concept. Here we only note that every Schaefer set is tight.

\begin{theorem}[\normalfont Gopalan et
  al.~\cite{GK1,GK2}]\label{GK-th1}\quad\\*
  Let~$\mathcal{S}$ be a finite set of logical relations.

  \qitem{(a)} If~$\mathcal{S}$ is tight, then
  \textsc{st-Conn($\mathcal{S}$)} is in~{\normalfont\textbf{P}}; otherwise,
  \textsc{st-Conn($\mathcal{S}$)} is {\normalfont\textbf{PSPACE}}-complete.

  \qitem{(b)} If~$\mathcal{S}$ is tight, then \textsc{Conn($\mathcal{S}$)}
  is in~{\normalfont\textbf{coNP}}; if it is tight but not Schaefer, then
  it is {\normalfont\textbf{coNP}}-complete; otherwise, it is
  {\normalfont\textbf{PSPACE}}-complete.

  \qitem{(c)} If every relation~$R$ in~$\mathcal{S}$ is the set of
  solutions of a 2-CNF-formula, then \textsc{Conn($\mathcal{S}$)} is
  in~{\normalfont\textbf{P}}.
\end{theorem}

\noindent
Major parts of the proof of Theorem~\ref{GK-th1} in~\cite{GK1,GK2} follow a
similar strategy to the proof of Schaefer's Theorem in~\cite{Sc}. Given a
set of relations~$\mathcal{S}$, a $k$-ary relation~$R$ is \emph{expressible
  from~$S$} if there is a CNF($\mathcal{S}$)-formula
$\varphi(x_1,\ldots,x_k,z_1,\ldots,z_m)$ such that~$R$ coincides with the
set of all assignments to $x_1,\ldots,x_k$ that satisfy
$(\exists z_1)\,\cdots\,(\exists z_m)\,
\varphi(x_1,\ldots,x_k,z_1,\ldots,z_m)$. Then the essential part of the
proof of Schaefer's Theorem is that if a set~$\mathcal{S}$ of relations is
not-Schaefer, then every logical relation is expressible
from~$\mathcal{S}$.

The authors in~\cite{GK1,GK2} extend the concept of expressibility to
\emph{structurally expressibility}. Informally, a relation~$R$ is
structurally expressible from a set of relations~$\mathcal{S}$, if~$R$ is
expressible using some CNF($\mathcal{S}$)-formula
$\varphi(x_1,\ldots,x_k,z_1,\ldots,z_m)$ and if the subgraph of the
hypercube formed by the satisfying assignments of~$\varphi$ has components
that `resemble' the components formed by the subgraph of the hypercube of
the relations in~$R$. The crucial result is that if a set of
relations~$\mathcal{S}$ is not tight, then every logical relation is
structurally expressible from~$\mathcal{S}$. Although the outline of this
(part of the) proof is very similar to the outline of the corresponding part
of Schaefer's Theorem, the actual proof is considerably more involved.

Additionally, the following results on the structure of the configuration
graphs of the solutions of different CNF($\mathcal{S}$)-formula are
obtained in~\cite{GK1,GK2}.

\begin{theorem}[\normalfont Gopalan et
  al.~\cite{GK1,GK2}]\label{GK-th2}\quad\\*
  Let~$\mathcal{S}$ be a finite set of logical relations.

  \qitem{(a)} If~$\mathcal{S}$ is tight, then for any
  CNF($\mathcal{S}$)-formula~$\varphi$ with~$n$ variables, if two
  satisfying assignments~$\textbf{s}$ and~$\textbf{t}$ of~$\varphi$ are
  connected by a path, then the number of transformations needed to go
  from~$\textbf{s}$ to~$\textbf{t}$ is $O(n)$.

  \qitem{(b)} If~$\mathcal{S}$ is not tight, then there exists an
  exponential function~$f(n)$ such that for every~$n_0$ there exists a
  CNF($\mathcal{S}$)-formula~$\varphi$ with $n\ge n_0$ variables and two
  satisfying assignments~$\textbf{s}$ and~$\textbf{t}$ of~$\varphi$ that
  are connected by a path, but where the number of transformations needed
  to go from~$\textbf{s}$ to~$\textbf{t}$ is at least~$f(n)$.
\end{theorem}

\noindent
Regarding the result in Theorem~\ref{GK-th1}\,(b), in their original
paper~\cite{GK1} the authors in fact conjectured a trichotomy for the
complexity of \textsc{Conn($\mathcal{S}$)}, conjecturing that
if~$\mathcal{S}$ is Schaefer, then \textsc{Conn($\mathcal{S}$)} is actually
in~$\textbf{P}$. They showed this conjecture to be true for particular
types of Schaefer set (see Theorem~\ref{GK-th1}\,(b) for one example). The
conjecture was disproved (assuming
$\text{\textbf{P}}\ne\text{\textbf{NP}}$) by Makino et al.~\cite{MT}, who
found a set of Schaefer relations for which \textsc{Conn($\mathcal{S}$)}
remains \textbf{coNP}-complete. In the updated version~\cite{GK2}, a
modified trichotomy conjecture for the complexity of
\textsc{Conn($\mathcal{S}$)} is formulated.

\section{Reconfiguration of graph colourings}\label{sec-col}

Reconfiguration of different kinds of graph colourings is probably one of
the best studied examples of reconfiguration problems. We will look at some
particular variants. The required background in the basics of graph theory
can be found in any textbook on graph theory, such as Bondy~\&
Murty~\cite{BM} or Diestel~\cite{Di}.

\subsection{Single-vertex recolouring of vertex colourings}

Recall that a \emph{$k$-colouring} of a graph $G=(V,E)$ is an assignment
$\varphi:V\rightarrow\{1,\ldots,k\}$ such that $\varphi(u)\ne\varphi(v)$
for every edge $uv\in E$. A graph is \emph{$k$-colourable} if it has a
$k$-colouring.

We start by considering the case where we are allowed to recolour one
vertex at a time, while always maintaining a valid $k$-colouring. We
immediately get the following two reconfiguration problems, for a fixed
positive integer~$k$.

\pagebreak[3]
\bigskip
\begin{probl}{$k$-Colour-Path}
\item[\normalfont\textit{Instance}:\hfill] A graph~$G$ together with
  two $k$-colourings~$\alpha$ and~$\beta$.
\item[\normalfont\textit{Question}:\hfill] Is it possible to transform the
  first colouring~$\alpha$ into the second colouring~$\beta$ by recolouring
  one vertex at a time, while always maintaining a valid $k$-colouring?
\end{probl}

\bigskip
\begin{probl}{$k$-Colour-Mixing}
\item[\normalfont\textit{Instance}:\hfill] A graph~$G$.
\item[\normalfont\textit{Question}:\hfill] Is it possible, for any two
  $k$-colourings of~$G$, to transform the first one into the second one by
  recolouring one vertex at a time, while always maintaining a valid
  $k$-colouring?
\end{probl}

\bigskip\noindent
Let us call a graph~$G$ \emph{$k$-mixing} if the answer to the second
decision problem is yes. The use of the work ``mixing'' in this context
derives from its relation with work on rapid mixing of Markov chains to
sample combinatorial configurations; we say more about this in the final
Section~\ref{sec-appl}.

A variant of graph colouring is \emph{list-colouring}. Here we assume that
each vertex~$v$ of a graph~$G$ has its own \emph{list}~$L(v)$ of colours.
An \emph{$L$-colouring} of the vertices is an assignment
$\varphi:V\longrightarrow\bigcup_{v\in V}L(v)$ such that
$\varphi(v)\in L(v)$ for each vertex $v\in V$, and
$\varphi(u)\ne\varphi(v)$ for every edge $uv\in E$. We call such a
colouring a \emph{$k$-list-colouring} if each list~$L(v)$ contains at
most~$k$ colours.

Regarding a transformation of a list-colouring, we use the obvious choice:
recolour one vertex at a time, where the new colour must come from the list
for that vertex list. With all this, we say that a graph~$G$ with list
assignments~$L$ is \emph{$L$-list-mixing} if for every two
$L$-colourings~$\alpha$ and~$\beta$, we can transform~$\alpha$ into~$\beta$
by recolouring one vertex at a time, while always maintaining a valid
$L$-colouring.

We also have two related decision problems.

\bigskip
\begin{probl}{$k$-List-Colour-Path}
\item[\normalfont\textit{Instance}:\hfill] A graph~$G$, list
  assignments~$L(v)$ with $|L(v)|\le k$ for each vertex $v\in V$, and two
  $L$-colourings~$\alpha$ and~$\beta$.
\item[\normalfont\textit{Question}:\hfill] Is it possible to transform the
  first colouring~$\alpha$ into the second colouring~$\beta$ by recolouring
  one vertex at a time, while always maintaining a valid $L$-colouring?
\end{probl}

\bigskip
\begin{probl}{$k$-List-Colour-Mixing}
\item[\normalfont\textit{Instance}:\hfill] A graph~$G$ and list
  assignments~$L(v)$ with $|L(v)|\le k$ for each vertex $v\in V$.
\item[\normalfont\textit{Question}:\hfill] Is~$G$ $L$-list-mixing?
\end{probl}

\bigskip\noindent
Before we look in some detail at what is known about the decision problems
defined in this subsection, we give some other results on (list-)mixing.
The following result has been obtained independently several times, but the
first instance appears to be in (preliminary versions) of Dyer et
al.~\cite{DF}. The \emph{degeneracy $\dege(G)$} of a graph~$G$ is the
minimum integer~$d$ so that every subgraph of~$G$ has a vertex of degree at
most~$d$. In other words, $\dege(G)$ is the maximum, over all subgraphs~$H$
of~$G$, of the minimum degree of~$H$.

\begin{theorem}[\normalfont Dyer et al.~\cite{DF}]\label{DF-th1}\quad\\*
  For any graph~$G$, if $k\ge\dege(G)+2$, then~$G$ is $k$-mixing. In fact,
  $G$ is $L$-list-mixing for any list assignment~$L$ such that
  $|L(v)|\ge\dege(G)+2$ for all $v\in V$.
\end{theorem}

\begin{proof}
  We only prove the $k$-mixing statement; the proof of the list version is
  essentially the same.

  We use induction on the number of vertices of~$G$. The result is
  obviously true for the graph with one vertex, so suppose~$G$ has at least
  two vertices. Let~$v$ be a vertex with degree at most $\dege(G)$, and
  consider $G'=G-v$. Note that $\dege(G')\le\dege(G)$, hence we also have
  $k\ge\dege(G')+2$, and by induction we can assume that~$G'$ is
  $k$-mixing.

  Take two $k$-colourings~$\alpha$ and~$\beta$ of~$G$, and let
  $\alpha',\beta'$ be the $k$-colourings of~$G'$ induced by $\alpha,\beta$.
  Since~$G'$ is $k$-mixing, there exists a sequence
  $\alpha'=\gamma'_0,\gamma'_1,\ldots,\gamma'_N=\beta'$ of $k$-colourings
  of~$G'$ so that two consecutive colourings $\gamma'_{i-1},\gamma'_i$,
  $i=1,\ldots,N$, differ in the colour of one vertex, say~$v_i$. Set
  $c_i=\gamma'_i(v_i)$, the colour of~$v_i$ after recolouring. We now try
  to take the same recolouring steps to recolour~$G$, starting
  from~$\alpha$. If for some~$i$ it is not possible to recolour~$v_i$, this
  must be because~$v_i$ is adjacent to~$v$ and~$v$ at that moment has
  colour~$c_i$. But because ~$v$ has degree at most $\dege(G)\le k-2$,
  there is a colour $c\ne c_i$ that does not appear on any of the
  neighbours of~$v$. Hence we can first recolour~$v$ to~$c$, then
  recolour~$v_i$ to~$c_i$, and continue.

  In this way we find a sequence of $k$-colourings of~$G$, starting
  at~$\alpha$, and ending in a colouring in which all the vertices except
  possibly~$v$ will have the same colour as in~$\beta$. But then, if
  necessary, we can do a final recolouring of~$v$ to give it the colour
  from~$\beta$, completing the proof.
\end{proof}

\noindent
Theorem~\ref{DF-th1} is best possible, as can be seen, e.g., by the
complete graphs~$K_n$ and trees. (For a tree~$T$ we have $\dege(T)=1$,
while it is trivial to check that a graph with at least one edge is never
2-mixing.) Constructions of graphs that are $k$-mixing for specific values
of~$k$, and not for other values, can be found in Cereceda et
al.~\cite{CHJ1}.

Since the degeneracy $\dege(G)$ of a graph~$H$ is clearly at most the
maximum degree~$\Delta(G)$, Theorem~\ref{DF-th1} immediately means that for
$k\ge\Delta(G)+2$, $G$ is $k$-mixing, as already noted by
Jerrum~\cite{Je1,Je2}.

An interesting result that is related to Theorem~\ref{DF-th1} was proved by
Choo~\& MacGillivray~\cite{CM}. They proved that if $k\ge\dege(G)+3$, then
the configuration graph (formed by all $k$-colourings with edges between
colourings that differ in the colour on one vertex) is Hamiltonian. In
other words, for those~$k$ we can start at any $k$-colouring of~$G$ and
then there is a sequence of single-vertex recolourings so that every other
$k$-colouring appears exactly once, ending with the original starting
colouring.

When $k\ge\dege(G)+2$, the proof of Theorem~\ref{DF-th1} provides an
algorithm to find a sequence of transformations between any two
$k$-colourings of~$G$. But the best upper bound on the number of steps that
can be obtained from the proof is exponential in the number of vertices
of~$G$. No graph is known for which such an exponential number of steps is
necessary. In fact the following is conjectured in Cereceda~\cite{Ce}.

\pagebreak
\begin{conj}[\normalfont Cereceda~\cite{Ce}]\label{Ce-con1}\quad\\*
  For a graph~$G$ on~$n$ vertices and integer $k\ge\dege(G)+2$, any two
  $k$-colourings of~$G$ can be transformed from one into the other using
  $O(n^2)$ single-vertex recolouring steps.
\end{conj}

\noindent
If true, the value $O(n^2)$ in Conjecture~\ref{Ce-con1} would be best
possible. Some weaker versions of the conjecture were proved in~\cite{Ce}.

\begin{theorem}[\normalfont Cereceda~\cite{Ce}]\label{Ce-th1}\quad\\*
  Conjecture~\ref{Ce-con1} is true under the stronger assumptions
  $k\ge2\dege(G)+1$ or $k\ge\Delta(G)+2$.
\end{theorem}

\noindent
Theorem~\ref{Ce-th1} means that Conjecture~\ref{Ce-con1} is true if~$G$ is
a tree (then $\dege(G)+2=3=2\dege(G)+1$) or if~$G$ is regular (in which
case $\dege(G)=\Delta(G)$).

Note that Theorem~\ref{DF-th1} has algorithmic consequences if we restrict
the decision problems to classes of graph in which each graph has
degeneracy at most some fixed upper bound. For instance, as planar graphs
have degeneracy at most~5, we obtain that \textsc{7-Colour-Mixing} and
\textsc{7-Colour-Path} restricted to planar graphs are trivially
in~\textbf{P}, as the answer is always ``yes''.

\medskip
We now return to the recolouring complexity problems introduced earlier in
this subsection. The following are some results that are known about those
problems.

\begin{theorem}\label{col-th1}\quad

  \qitem{(a)} If $k=2$, then \textsc{$k$-Colour-Path} and
  \textsc{$k$-Colour-Mixing} are in~{\normalfont\textbf{P}}.

  \qitem{(b)} If $k=3$, then \textsc{$k$-Colour-Path} is
  in~{\normalfont\textbf{P}}; while \textsc{$k$-Colour-Mixing} is
  {\normalfont\textbf{coNP}}-complete (Cereceda et al.~\cite{CHJ2,CHJ3}).

  \qitem{(c)} For all $k\ge4$, \textsc{$k$-Colour-Path} is
  {\normalfont\textbf{PSPACE}}-complete (Bonsma and Cereceda~\cite{BC}).
\end{theorem}

\begin{theorem}\label{col-th2}\quad

  \qitem{(a)} If $k=2$, then \textsc{$k$-List-Colour-Path} and
  \textsc{$k$-List-Colour-Mixing} are in~{\normalfont\textbf{P}}.

  \qitem{(b)} For all $k\ge3$, \textsc{$k$-List-Colour-Path} is
  {\normalfont\textbf{PSPACE}}-complete (Bonsma and Cereceda~\cite{BC}).
\end{theorem}

\noindent
We already noted that the claims in Theorem~\ref{col-th1}\,(a) are trivial:
if we have only two colours, then the end-vertices of any edge can never be
recoloured.

The results in Theorem~\ref{col-th1}\,(b) are clearly the odd ones among
the list of results above. The proof in~\cite{CHJ2} of the
\textbf{coNP}-completeness of \textsc{3-Colour-Mixing} uses the concept of
\emph{folding}: given two non-adjacent vertices~$u$ and~$v$ that have a
common neighbour, a \emph{fold on~$u$ and~$v$} is the identification of~$u$
and~$v$ (together with removal of any double edges produced). We say a
graph~$G$ is \emph{foldable} to~$H$ if there exists a sequence of folds
that transforms~$G$ to~$H$. Folding of graphs, and its relation to vertex
colouring, has been studied before, see for instance~\cite{CE}.

Combining observations from~\cite{CHJ1} for non-bipartite 3-colourable
graphs, with the structural characterisation of bipartite graphs that are
not 3-mixing in~\cite{CHJ2}, gives the following.

\begin{theorem}[\normalfont Cereceda et
  al.~\cite{CHJ1,CHJ2}]\label{th-fold}\quad\\*
  Let~$G$ be a connected 3-colourable graph. Then~$G$ is not 3-mixing if
  and only if~$G$ is foldable to the 3-cycle~$C_3$ or the 6-cycle~$C_6$.  
\end{theorem}

\noindent
It is easy to see that every non-bipartite 3-colourable connected graph can
be folded to~$C_3$, so the interesting part of this theorem is the
characterisation of bipartite non-3-mixing graphs as being foldable
to~$C_6$. In a sense, the theorem shows that~$C_3$ and~$C_6$ are the
`minimal' graphs that are not 3-mixing.

\medskip
Theorem~\ref{col-th2}\,(b) is not explicitly given in Bonsma and
Cereceda~\cite{BC}, but follows from the proof of
Theorem~\ref{col-th1}\,(c) in that paper.

The results that \textsc{2-List-Colour-Path} and
\textsc{2-List-Colour-Mixing} can be done in polynomial time can be proved
directly with some effort. But it is easy to see that in fact any
2-list-colouring problem can be reduced to a Boolean 2-CNF-formula. For
each vertex~$v$ and colour $c\in L(v)$, introduce a Boolean
variable~$x_{v,c}$. Then for each vertex~$v$ with $L(v)=\{c,d\}$ we add a
clause $x_{v,c}\vee x_{v,d}$; while for each edge~$uv$ and colour
$c\in L(u)\cap L(v)$ we add a clause $(\neg x_{u,c})\vee(\neg x_{v,c})$. We
can now use the results that checking the connectivity of the solution
space of a 2-CNF-formula is in~\textbf{P}, see Theorem~\ref{GK-th1}\,(c).

\medskip
As well as the computational complexity of some of the recolouring
problems, we also know something about the number of recolourings we might
need.

\begin{theorem}\label{col-th3}\quad

  \qitem{(a)} For a graph~$G$ on~$n$ vertices, if two 3-colourings of~$G$
  can be connected by a sequence of single-vertex recolourings, then this
  can be done in $O(n^2)$ steps (Cereceda et al.~\cite{CHJ3}).

  \qitem{(b)} For any $k\ge3$ there exists an exponential function~$f(n)$
  such that for every~$n_0$, there exists a graph~$G$ on $n\ge n_0$
  vertices, an assignment~$L$ of lists of size~$k$ to the vertices of~$G$,
  and two $L$-colourings~$\alpha$ and~$\beta$ of~$G$, such that we can
  transform~$\alpha$ into~$\beta$ by recolouring one vertex at a time, but
  where the number of recolourings required is at least~$f(n)$ (Bonsma and
  Cereceda~\cite{BC}).

  \qitem{(c)} For any $k\ge4$, the result in~(b) also holds for ordinary
  $k$-colouring recolourings.
\end{theorem}

\noindent
The bound $O(n^2)$ in Theorem~\ref{col-th3}\,(a) is best possible.

Two remaining questions are the complexity class of
\textsc{$k$-Colour-Mixing} for $k\ge4$ and \textsc{$k$-List-Colour-Mixing}
for $k\ge3$. In view of the fact that \textsc{$k$-Colour-Path} and
\textsc{$k$-List-Colour-Path} for those values of~$k$ are
\textbf{PSPACE}-complete, one would expect the mixing variants to have a
similar complexity. On the other hand, since the \textbf{coNP}-completeness
of \textsc{3-Colour-Mixing} is obtained by a particular structure that
needs to be present in a graph to fail to be 3-mixing, a similar
graph-structural condition for mixing with more colours might well be
possible.

In view of the results that many of the recolouring problems are not
in~\textbf{P} (assuming that $\text{\textbf{P}}\ne\text{\textbf{NP}}$), it
is interesting to find restricted instances for which the recognition
problems are in~\textbf{P}. A natural choice for a graph class where one
expects this to happen is the class of bipartite graphs, since many
colouring problems are trivial in that class. Surprisingly, restricting the
input of the decision problems to just bipartite graphs does not change any
of the results in Theorems \ref{col-th1}--\ref{col-th3}. A restriction to
planar graphs has more surprising results, as expressed in the final result
of this subsection.

\begin{theorem}\label{col-th4}\quad

  \qitem{(a)} When restricted to planar graphs, \textsc{3-Colour-Mixing}
  becomes polynomial (Cereceda et al.~\cite{CHJ2}).

  \qitem{(b)} When restricted to planar graphs, \textsc{$k$-Colour-Path}
  for $4\le k\le6$ and \textsc{$k$-List-Colour-Path} for $3\le k\le6$
  remain {\normalfont\textbf{PSPACE}}-complete (Bonsma and
  Cereceda~\cite{BC}).\\
  Both decision problems are in~{\normalfont\textbf{P}} for $k\ge7$ when
  restricted to planar graphs.

  \qitem{(c)} When restricted to bipartite planar graphs,
  \textsc{$k$-Colour-Path} for $k=4$ and \textsc{$k$-List-Colour-Path} for
  $3\le k\le4$ remain {\normalfont\textbf{PSPACE}}-complete (Bonsma and
  Cereceda~\cite{BC}).\\
  Both decision problems are in~{\normalfont\textbf{P}} for $k\ge5$ when
  restricted to bipartite planar graphs.
\end{theorem}

\noindent
The results that \textsc{$k$-Colour-Path} and \textsc{$k$-List-Colour-Path}
are polynomial for planar or bipartite planar graphs and larger~$k$,
follows directly from upper bounds for the degeneracy for those graphs and 
Theorem~\ref{DF-th1}.

\subsection{Kempe chain recolouring}\label{subs-Kc}

Given a $k$-colouring~$\varphi$ of a graph~$G$, a \emph{Kempe chain} is a
connected component of the subgraph of~$G$ induced by the vertices coloured
with one of two give colours. In other words, if $c_1,c_2$ are two
different colours, and $W\subseteq V$ is the collection of vertices
coloured either~$c_1$ or~$c_2$, then a Kempe chain is a connected component
of the induced subgraph of~$G$ with vertex set~$W$. By a \emph{Kempe chain
  recolouring} we mean switching the two colours on a Kempe chain. Kempe
chains and Kempe chain recolouring have been essential concepts in the
proofs of many classical results on colouring, such as the Four Colour
Theorem from Appel~\& Haken~\cite{AH1,AH2,AH3} and Vizing's Edge-Colouring
Theorem~\cite{Vi}.

Notice that a Kempe chain recolouring is a generalisation of the
single-vertex recolouring transformation from the previous subsection,
since such a recolouring corresponds to a Kempe chain recolouring on a
Kempe chain consisting of just one vertex. Let us call a graph~$G$
\emph{$k$-Kempe-mixing} if it is possible, for any two $k$-colourings
of~$G$, to transform the first one into the second one by a sequence of
Kempe chain recolourings.

From the observation above, we see that if a graph is $k$-mixing, then it
certainly is $k$-Kempe-mixing. But the reverse need not be true. For
instance, it has been observed many times that a bipartite graph is
$k$-Kempe-mixing for any $k\ge2$ \cite{BH,FS,Mo}, whereas for any $k\ge2$
there exist bipartite graphs that are not $k$-mixing~\cite{CHJ1}.
Furthermore, a simple modification of the proof of Theorem~\ref{DF-th1}
shows that for any graph~$G$, if $k\ge\dege(G)+1$, then~$G$ is
$k$-Kempe-mixing, as was already proved by Las Vergnas~\&
Meyniel~\cite{LM}.

Very little is known about the complexity of determining if a graph is
$k$-Kempe-mixing. The same holds for the `path' version of the problem
(determining if two given $k$-colourings can be transformed into one
another by a sequence of Kempe chain recolourings). Intuitively, there
appear to be at least two reasons why Kempe chain recolouring is so much
harder to analyse than single-vertex recolouring. Firstly, for any
$k$-colouring of a graph it is always possible to perform Kempe chain
recolourings. This is different from single-vertex recolouring, where it
might not be possible to recolour any vertex at all (if all vertices have
all colours different from their own appearing on a neighbour). This kind
of `frozen' colourings is essential in many of the analyses and results on
single-vertex recolourings. Secondly, whereas a single-vertex recolouring
has only a `local' effect, a Kempe chain can affect many vertices
throughout the graph.

The following are some results for planar graphs.

\begin{theorem}\label{Ke-th1}\quad
  
  \qitem{(a)} Every planar graph is 5-Kempe-mixing (Meyniel~\cite{Me}).

  \qitem{(b)} If~$G$ is a 3-colourable planar graph, then~$G$ is
  4-Kempe-mixing (Mohar~\cite{Mo}).
\end{theorem}

\noindent
The results in the theorem are best possible in the sense that the number
of colours cannot be reduced in either statement~\cite{Mo}.

As we have observed earlier, for the single-vertex recolouring problem, the
smallest graph that is not 3-mixing is the triangle~$C_3$, while the
smallest \emph{bipartite} graph that is not 3-mixing is the 6-cycle~$C_6$.
These graphs are essential in the proof that \textsc{3-Colour-Mixing} is
\textbf{coNP}-complete. The smallest graph that is not 3-Kempe-mixing is
the 3-prism $K_3\mathop{\Box}K_2$ (see Figure~\ref{fig:prism}).
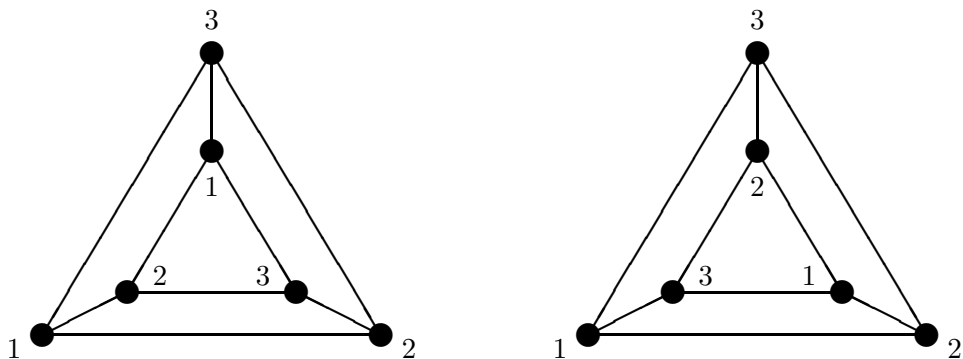
\begin{figure}[ht]
  \medskip
  \centering
  \unitlength1.25mm
  \begin{picture}(44,38)(-4,-4)
    \put(0,0){\circle*{2.5}}
    \put(36,0){\circle*{2.5}}
    \put(18,30){\circle*{2.5}}
    \put(9,4.5){\circle*{2.5}}
    \put(27,4.5){\circle*{2.5}}
    \put(18,19.5){\circle*{2.5}}
    \thicklines
    \put(0,0){\line(1,0){36}}
    \put(18,30){\line(-3,-5){18}}
    \put(18,30){\line(3,-5){18}}
    \put(0,0){\line(2,1){9}}
    \put(36,0){\line(-2,1){9}}
    \put(18,30){\line(0,-1){10.5}}
    \put(9,4.5){\line(1,0){18}}
    \put(18,19.5){\line(-3,-5){9}}
    \put(18,19.5){\line(3,-5){9}}
    \put(-3,-1.5){\makebox(0,0){1}}
    \put(39,-1.5){\makebox(0,0){2}}
    \put(18,33.5){\makebox(0,0){3}}
    \put(12.5,6.25){\makebox(0,0){2}}
    \put(23.5,6.25){\makebox(0,0){3}}
    \put(18,15.75){\makebox(0,0){1}}
  \end{picture}
  \hspace{15mm}
  \begin{picture}(44,38)(-4,-4)
    \put(0,0){\circle*{2.5}}
    \put(36,0){\circle*{2.5}}
    \put(18,30){\circle*{2.5}}
    \put(9,4.5){\circle*{2.5}}
    \put(27,4.5){\circle*{2.5}}
    \put(18,19.5){\circle*{2.5}}
    \thicklines
    \put(0,0){\line(1,0){36}}
    \put(18,30){\line(-3,-5){18}}
    \put(18,30){\line(3,-5){18}}
    \put(0,0){\line(2,1){9}}
    \put(36,0){\line(-2,1){9}}
    \put(18,30){\line(0,-1){10.5}}
    \put(9,4.5){\line(1,0){18}}
    \put(18,19.5){\line(-3,-5){9}}
    \put(18,19.5){\line(3,-5){9}}
    \put(-3,-1.5){\makebox(0,0){1}}
    \put(39,-1.5){\makebox(0,0){2}}
    \put(18,33.5){\makebox(0,0){3}}
    \put(12.5,6.25){\makebox(0,0){3}}
    \put(23.5,6.25){\makebox(0,0){1}}
    \put(18,15.75){\makebox(0,0){2}}
  \end{picture}
  \medskip
  \caption{The 3-prism with two 3-colourings that are not related by Kempe
    chains}
  \label{fig:prism}
\end{figure}

It is easy to check that any Kempe chain recolouring in these two
colourings will only result in renaming the two colours involved in the
Kempe chain, but never changes the structure. It is unknown if the 3-prism
is in some way a `minimal' graph that is not 3-Kempe mixing, or if it is
the only one. Neither is it obvious what subgraph relation we should use
(like `foldable' for the single-vertex recolouring problem) when talking
about `minimal' for Kempe-mixing.

\medskip
Kempe chains have also been used extensively in the analysis of
edge-colourings of graphs. Recall that a \emph{$k$-edge-colouring} of a
graph $G=(V,E)$ is an assignment $\varphi:E\rightarrow\{1,\ldots,k\}$ such
that $\varphi(e)\ne\varphi(f)$ for any two edges~$e,f$ that share a common
end-vertex. Similar to vertex-colourings, a Kempe chain in an edge-coloured
graph is a component of the subgraph formed by the edges coloured with one
of two given colours. Note that in this case every Kempe chain is a path or
an even length cycle, and a recolouring is again just switching the colours
on the chain. Call a graph~$G$ \emph{$k$-Kempe-edge-mixing} if it is
possible for any two $k$-edge-colourings of~$G$ to transform the first one
into the second one by a sequence of Kempe chain recolourings.

Kempe chains on edge-colourings are instrumental in most (if not all)
proofs of Vizing's Theorem~\cite{Vi} that a simple graph with maximum
degree~$\Delta$ has an edge-colouring using at most $\Delta+1$ colours.
Hence it is not surprising that results on $k$-Kempe-edge mixing are
related to this constant as well.

\begin{theorem}[\normalfont Mohar~\cite{Mo}]\label{Mo-th1}\quad

  \qitem{(a)} If a simple graph~$G$ can be edge-coloured with~$k$ colours,
  then~$G$ is $(k+2)$-Kempe-edge-mixing.

  \qitem{(b)} If~$G$ is a simple bipartite graph with maximum
  degree~$\Delta$, then~$G$ is $(\Delta+1)$-Kempe-edge-mixing.
\end{theorem}

\noindent
It is unknown if Theorem~\ref{Mo-th1}\,(a) is best possible, nor if the
condition that the graph be bipartite in part~(b) is necessary. A strongest
possible result would be for any simple graph with maximum degree~$\Delta$
to be $(\Delta+1)$-Kempe-edge-mixing.

\section{Moving tokens on graphs}\label{sec-tok}

The 15-puzzle can be considered as a problem involving moving tokens around
a given graph, where a token can be moved along an edge to an empty vertex.
So the two configurations in Figure~\ref{fig:15puzzle} can also be drawn as
in Figure~\ref{fig:15graph}.
\begin{figure}[ht]
  \medskip
  \centering
  \unitlength.6mm
  \begin{picture}(68,68)(-4,-4)
    \put(0,0){\circle{9}}\put(0,0){\makebox(0,0){15}}
    \put(20,0){\circle{9}}\put(20,0){\makebox(0,0){8}}
    \put(40,0){\circle{9}}\put(40,0){\makebox(0,0){7}}
    \put(60,0){\circle{9}}\put(60,0){\makebox(0,0){5}}
    \put(0,20){\circle*{4}}
    \put(20,20){\circle{9}}\put(20,20){\makebox(0,0){6}}
    \put(40,20){\circle{9}}\put(40,20){\makebox(0,0){4}}
    \put(60,20){\circle{9}}\put(60,20){\makebox(0,0){14}}
    \put(0,40){\circle{9}}\put(0,40){\makebox(0,0){9}}
    \put(20,40){\circle{9}}\put(20,40){\makebox(0,0){11}}
    \put(40,40){\circle{9}}\put(40,40){\makebox(0,0){1}}
    \put(60,40){\circle{9}}\put(60,40){\makebox(0,0){10}}
    \put(0,60){\circle{9}}\put(0,60){\makebox(0,0){13}}
    \put(20,60){\circle{9}}\put(20,60){\makebox(0,0){2}}
    \put(40,60){\circle{9}}\put(40,60){\makebox(0,0){3}}
    \put(60,60){\circle{9}}\put(60,60){\makebox(0,0){12}}
    \thicklines
    \put(4.5,0){\line(1,0){11.1}}
    \put(24.5,0){\line(1,0){11.1}}
    \put(44.5,0){\line(1,0){11.1}}
    \put(0,20){\line(1,0){15.6}}
    \put(24.5,20){\line(1,0){11.1}}
    \put(44.5,20){\line(1,0){11.1}}
    \put(4.5,40){\line(1,0){11.1}}
    \put(24.5,40){\line(1,0){11.1}}
    \put(44.5,40){\line(1,0){11.1}}
    \put(4.5,60){\line(1,0){11.1}}
    \put(24.5,60){\line(1,0){11.1}}
    \put(44.5,60){\line(1,0){11.1}}
    \put(0,4.5){\line(0,1){31.2}}
    \put(20,4.5){\line(0,1){11.1}}
    \put(40,4.5){\line(0,1){11.1}}
    \put(60,4.5){\line(0,1){11.1}}
    \put(20,24.5){\line(0,1){11.1}}
    \put(40,24.5){\line(0,1){11.1}}
    \put(60,24.5){\line(0,1){11.1}}
    \put(0,44.5){\line(0,1){11.1}}
    \put(20,44.5){\line(0,1){11.1}}
    \put(40,44.5){\line(0,1){11.1}}
    \put(60,44.5){\line(0,1){11.1}}
  \end{picture}
  \hspace{15mm}
  \begin{picture}(68,68)(-4,-4)
    \put(0,0){\circle{9}}\put(0,0){\makebox(0,0){13}}
    \put(20,0){\circle{9}}\put(20,0){\makebox(0,0){14}}
    \put(40,0){\circle{9}}\put(40,0){\makebox(0,0){15}}
    \put(60,0){\circle*{4}}
    \put(0,20){\circle{9}}\put(0,20){\makebox(0,0){9}}
    \put(20,20){\circle{9}}\put(20,20){\makebox(0,0){10}}
    \put(40,20){\circle{9}}\put(40,20){\makebox(0,0){11}}
    \put(60,20){\circle{9}}\put(60,20){\makebox(0,0){12}}
    \put(0,40){\circle{9}}\put(0,40){\makebox(0,0){5}}
    \put(20,40){\circle{9}}\put(20,40){\makebox(0,0){6}}
    \put(40,40){\circle{9}}\put(40,40){\makebox(0,0){7}}
    \put(60,40){\circle{9}}\put(60,40){\makebox(0,0){8}}
    \put(0,60){\circle{9}}\put(0,60){\makebox(0,0){1}}
    \put(20,60){\circle{9}}\put(20,60){\makebox(0,0){2}}
    \put(40,60){\circle{9}}\put(40,60){\makebox(0,0){3}}
    \put(60,60){\circle{9}}\put(60,60){\makebox(0,0){4}}
    \thicklines
    \put(4.5,0){\line(1,0){11.1}}
    \put(24.5,0){\line(1,0){11.1}}
    \put(44.5,0){\line(1,0){15.6}}
    \put(4.5,20){\line(1,0){11.1}}
    \put(24.5,20){\line(1,0){11.1}}
    \put(44.5,20){\line(1,0){11.1}}
    \put(4.5,40){\line(1,0){11.1}}
    \put(24.5,40){\line(1,0){11.1}}
    \put(44.5,40){\line(1,0){11.1}}
    \put(4.5,60){\line(1,0){11.1}}
    \put(24.5,60){\line(1,0){11.1}}
    \put(44.5,60){\line(1,0){11.1}}
    \put(0,4.5){\line(0,1){11.1}}
    \put(20,4.5){\line(0,1){11.1}}
    \put(40,4.5){\line(0,1){11.1}}
    \put(60,0){\line(0,1){15.6}}
    \put(0,24.5){\line(0,1){11.1}}
    \put(20,24.5){\line(0,1){11.1}}
    \put(40,24.5){\line(0,1){11.1}}
    \put(60,24.5){\line(0,1){11.1}}
    \put(0,44.5){\line(0,1){11.1}}
    \put(20,44.5){\line(0,1){11.1}}
    \put(40,44.5){\line(0,1){11.1}}
    \put(60,44.5){\line(0,1){11.1}}
  \end{picture}
  \medskip
  \caption{Two configurations of the 15-puzzle on the $4\times 4$ grid}
  \label{fig:15graph}
\end{figure}
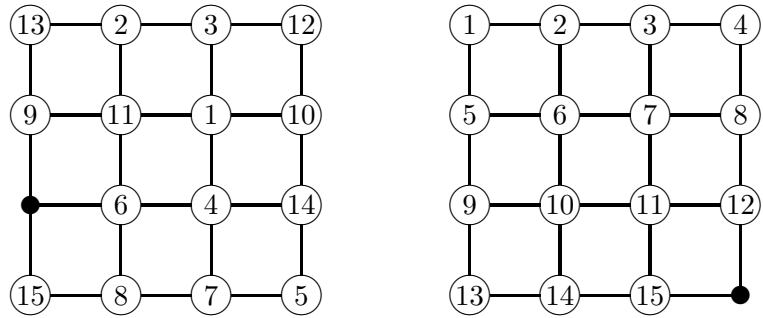

Looking at the 15-puzzle in this way immediately suggests all kind of
generalisations. An obvious generalisation is to play the game on different
graphs. But we can also change the number of tokens, or the way the tokens
are labelled. In this section we consider some of the variants that have
been studied in the literature.

\subsection{Labelled tokens without restrictions}

There is an obvious generalisation of the 15-puzzle. For a given graph
on~$n$ vertices, place $n-1$ tokens labelled 1 to $n-1$ on different
vertices. The allowed moves are ``sliding'' a token along an edge onto the
unoccupied vertex. The central question is if each of the $n!$ possible
token configurations can be obtained from one another by a sequence of
token moves. A complete answer to this was given in Wilson~\cite{Wi}. For a
given graph~$G$ on~$n$ vertices, he defines the \emph{puzzle graph
  $\puz(G)$} as the graph that has as vertex set all possible placements of
the $n-1$ tokens on~$G$, and two configurations are adjacent if they can be
obtained from one another by a single move.

\begin{theorem}[\normalfont Wilson~\cite{Wi}]\label{Wi-th1}\quad\\*
  Let~$G$ be a 2-connected graph on $n\ge3$ vertices. Then $\puz(G)$ is
  connected, except in the following cases:

  \qitem{(a)} $G$ is a cycle on $n\ge4$ vertices (in which case $\puz(G)$
  has $(n-2)!$ components);

  \qitem{(b)} $G$ is a bipartite graph different from a cycle (then
  $\puz(G)$ has two components);

  \qitem{(c)} $G$ is the graph~$\theta_0$ in Figure~\ref{fig:th0} (then
  $\puz(G)$ has six components).

  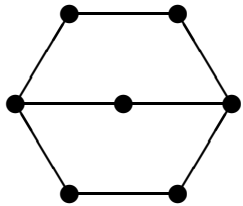
\begin{figure}[ht]
    \medskip
    \centering
    \unitlength0.8mm
    \begin{picture}(40,34)(-2,-2)
      \put(9,0){\circle*{3}}
      \put(27,0){\circle*{3}}
      \put(0,15){\circle*{3}}
      \put(18,15){\circle*{3}}
      \put(36,15){\circle*{3}}
      \put(9,30){\circle*{3}}
      \put(27,30){\circle*{3}}
      \thicklines
      \put(9,0){\line(1,0){18}}
      \put(9,0){\line(-3,5){9}}
      \put(27,0){\line(3,5){9}}
      \put(0,15){\line(1,0){36}}
      \put(0,15){\line(3,5){9}}
      \put(36,15){\line(-3,5){9}}
      \put(9,30){\line(1,0){18}}
    \end{picture}
    \medskip
    \caption{The exceptional graph~$\theta_0$}
    \label{fig:th0}
  \end{figure}
\end{theorem}

\noindent
The condition in Wilson's theorem that the graph~$G$ is 2-connected is
necessary. It is obvious that for a non-connected graph~$G$, $\puz(G)$ is
never connected; while if~$G$ has a cut-vertex~$v$, then a token can never
be moved from one component of $G-v$ to another component.

The proof of Theorem~\ref{Wi-th1} in~\cite{Wi} is quite algebraic in
nature. This is not surprising, since each token configuration can be
considered as a permutation of~$n$ labels (with the unoccupied vertex
having the label `empty'). Within that context, a move is just a particular
type of transposition involving two labels (one of them always being the
`empty' label). Although Wilson's theorem is not formulated in algorithmic
terms, it is easy to derive from it a polynomial time algorithm to decide
if $\puz(G)$ is connected for a given input graph~$G$.

Since Wilson's work (and often independent of it), many generalisations
have been considered in the literature. To describe these in some detail,
we need some further notation. Instead of assuming that all tokens are
different, we will assume that some tokens can be identical. So tokens come
in certain \emph{types} (other authors use \emph{colours} for this), where
tokens of the same type are considered indistinguishable (and hence
swapping tokens of the same type will not lead to a different
configuration). A collection of tokens can have~$k_1$ tokens of type~1,
$k_2$ tokens of type~2, etc. We denote such a typed set by
$(k_1,k_2,\ldots,k_p)$, so that $k_1+\dots+k_p$ is the total number of
tokens. A repeated sequence of~$p$ ones can be denoted as~$1^{(p)}$.

Given a graph~$G$ and token set $(k_1,\ldots,k_p)$, the \emph{puzzle graph
  $\puz(G;k_1,\dots,k_p)$} is the graph that has as vertex set all possible
token placements on~$G$ of~$k_1$ tokens of type~1, $k_2$ tokens of type~2,
etc., and two configurations are adjacent if they can be obtained from one
another by a single move of a token to a neighbouring empty vertex. This
means that if~$G$ is a graph on~$n$ vertices, then
$\puz(G)\cong\puz(G;1^{(n-1)})$. We will always assume that if~$G$ has~$n$
vertices, then $k_1+\dots+k_p\le n$ and $k_1\ge k_2\ge\dots\ge k_p\ge1$.

A first generalisation of Wilson's work, in which there may be fewer than
$n-1$ tokens, was considered by Kornhauser et al.~\cite{KM}. They showed
that if~$G$ is a graph on~$n$ vertices, then for any two configurations
from $\puz(G,1^{(p)})$, it can be decided in polynomial time if these two
configurations are in the same component, i.e., if one configuration can be
obtained from the other by a sequence of token moves. Additionally, they
showed that if such a transformation is possible, the number of moves
required is at most $O(n^3)$, and the order of this bound is best possible.

This work was further extended to token configurations with types as above;
first to trees by Auletta et al.~\cite{AM}, and later to general graphs by
Goraly~\& Hassin~\cite{GH}. Their results prove that for any graph~$G$ and
typed token set $(k_1,\ldots,k_p)$, given two configurations from
$\puz(G;k_1,\ldots,k_p)$, it can be decided in linear time if one
configuration can be obtained from the other. Notice that by the result for
all tokens being different mentioned earlier, we immediately have that for
a graph on~$n$ vertices, more than $O(n^3)$ token moves are never needed
between two configurations.

The work mentioned in the previous paragraphs does not give an explicit
characterisation of the puzzle graphs $\puz(G;k_1,\ldots,k_p)$ that are
connected (i.e., where any two token configurations of the right type can
be obtained from one another by a sequence of token moves). In order to
describe such a characterisation, we need some further terminology
regarding specific vertex-cut-sets in graphs. For a connected graph~$G$, a
\emph{separating path of size one} in~$G$ is a cut-vertex. A
\emph{separating path of size two} is a cut-edge $e=v_1v_2$ so that both
components of $G-e$ have at least two vertices. Finally, for $\ell\ge3$, a
\emph{separating path of size~$\ell$} is a path $P=v_1v_2\ldots v_\ell$
in~$G$, such that the vertices $v_2,\ldots,v_{\ell-1}$ have degree two,
$G-\{v_2,\ldots,v_{\ell-1}\}$ has exactly two components, one
containing~$v_1$ and one containing~$v_\ell$, and where both components
have at least two vertices.

\begin{theorem}[\normalfont Brightwell et
  al.~\cite{BHT}]\label{BvdHS-th1}\quad\\*
  Let~$G$ be a graph on $n\ge3$ vertices and $(k_1,\ldots,k_p)$ be a token
  set, with $k_1+\dots+k_p\le n$ and $k_1\ge k_2\ge\dots\ge k_p\ge1$. Then
  $\puz(G;k_1,\ldots,k_p)$ is disconnected if and only if at least one of
  the following cases holds:

  \qitem{(a)} $G$ is disconnected, and $p\ge2$ or $k_1\le n-1$;

  \qitem{(b)} $p\ge2$ and $k_1+\dots+k_p=n$;

  \qitem{(c)} $G$ is a path and $p\ge2$;

  \qitem{(d)} $G$ is a cycle, $p=2$ and $k_2\ge2$; or~$G$ is a cycle
  and $p\ge3$;

  \qitem{(e)} $G$ is a 2-connected bipartite graph and the token set is
  $(1^{(n-1)})$;

  \qitem{(f)} $G$ is the graph~$\theta_0$ in Figure~\ref{fig:th0} and the
  token set is one of $(2,2,2)$, $(2,2,1,1)$, $(2,1,1,1,1)$, $(1^{(6)})$;

  \qitem{(g)} $G$ has connectivity one and contains a separating path of
  size at least $n-(k_1+k_2+\dots+k_p)$.
\end{theorem}

\noindent
Note in particular that if~$G$ is a 2-connected graph on~$n$ vertices
different from a cycle, and $(k_1,\ldots,k_p)$ is a token set with
$k_1+\dots+k_p\le n-2$, then $\puz(G;k_1,\ldots,k_p)$ is always connected.

It is possible to extend this theorem to a full characterisation of any two
token configurations from any puzzle graph $\puz(G;k_1,\ldots,k_p)$ that
are in the same component (hence extending the algorithmic results from
Goraly~\& Hassin~\cite{GH}). This rather technical and long result can also
be found in~\cite{BHT}.

The results mentioned above mean that it is quite straightforward to check
if one can go from any given token configuration to any other one. So a
next natural question is to ask if it is possible to find the shortest
path, i.e., to find the minimum number of token moves required between two
given token configurations in the same component of
$\puz(G;k_1,\ldots,k_p)$. This leads to the following decision problem.

\bigskip
\begin{probl}{Shortest-Token-Moves-Sequence}
\item[\normalfont\textit{Instance}:\hfill] A graph~$G$, a token set
  $(k_1,\ldots,k_p)$, two token configurations~$\alpha$ and~$\beta$ on~$G$
  of type $(k_1,\ldots,k_p)$, and a positive integer~$N$.
\item[\normalfont\textit{Question}:\hfill] Is it possible to transform
  configuration~$\alpha$ into configuration~$\beta$ using at most~$N$ token
  moves?
\end{probl}

\begin{theorem}\label{short-th1}\quad\\*
  Restricted to the case that the token sets are~$(k)$ (i.e., all tokens
  are the same), \textsc{Shortest-Token-Moves-Sequence} is
  in~{\normalfont\textbf{P}}.
\end{theorem}

\begin{proof}
  We can assume that the given graph~$G$ is connected. (Since two
  configurations can be transformed into one another if and only if this
  can be done for the configurations restricted to the components of the
  graph.) Given two token configurations~$\alpha$ and~$\beta$ of~$k$
  identical tokens on~$G$, let $U=\{u_1,\ldots,u_k\}$ be the set of
  vertices containing a token in~$\alpha$, and $V=\{v_1,\ldots,v_k\}$ be
  the same for~$\beta$.

  Form a complete bipartite graph~$K_{k,k}$ with parts~$U$ and~$V$. For
  each edge $e_{ij}=u_iv_j$, define the weight~$w_{ij}$ of~$e_{ij}$ as the
  length of the shortest path from~$u_i$ to~$v_j$ in~$G$ (and denote
  by~$P_{ij}$ such a shortest path in~$G$). It is well-known that a minimum
  weight perfect matching in a weighted balanced complete bipartite graph
  can be found in polynomial time (for instance using the Hungarian method,
  see, e.g., Schrijver~\cite[Section~17.2]{Sch}); let~$M$ be such a minimum
  weight perfect matching.

  We can assume that $M=\{u_1v_1,\ldots,u_kv_k\}$. Let~$W$ be the total
  weight in~$M$, i.e., the sum of the lengths of the paths~$P_{ii}$,
  $i=1,\ldots,k$. It is obvious that any way to move the tokens from~$U$
  to~$V$ will use at least~$W$ steps. We will prove that in fact it is
  possible to do so using exactly~$W$ steps. We use induction on~$W$,
  observing that if $W=0$, then $U=V$, so $\alpha\equiv\beta$, and no
  tokens have to be moved.

  If $W>0$, then at least one element of~$V$, say~$v_1$, has no token on it
  in~$\alpha$. If $V(P_{11})\cap U=\{u_1\}$, then we can just move the
  token from~$u_1$ along~$P_{11}$ to~$v_1$, and are done by induction. So
  assume that~$P_{11}$ contains some other elements from~$U$. Take~$u_i$ to
  be the element from $V(P_{11})\cap U$ nearest to~$v_1$ on~$P_{11}$.
  Define new paths~$P'_{1i}$ and~$P'_{i1}$ as follows. Let~$P'_{1i}$ be the
  path formed by going from~$u_1$ along~$P_{11}$ to~$u_i$ and then continue
  along~$P_{ii}$ to~$v_i$; while~$P'_{i1}$ is just the path from~$u_i$
  along~$P_{11}$ to~$v_1$. It is clear that the sum of the lengths
  of~$P'_{1i}$ and~$P'_{i1}$ is the same as that sum for~$P_{11}$
  and~$P_{ii}$, so we can replace~$P_{11}$ and~$P_{ii}$ by~$P'_{1i}$
  and~$P'_{i1}$ to get another set of paths from~$U$ to~$V$ of minimum
  total length. But in this new collection of paths, we can just move~$u_i$
  along~$P'_{i1}$ to~$v_1$, and then continue by induction.
\end{proof}

\noindent
It was proved by Goldreich~\cite{Go}\footnote{Although~\cite{Go} was
  published in 2011, it is remarked in it that the work was already
  completed in 1984, and appeared as a technical report from the Technion
  in 1993.} that \textsc{Shortest-Token-Moves-Sequence} is
\textbf{NP}-complete for the case Wilson considered, i.e., if all tokens
are different. So somewhere between all tokens the same and all tokens
different, the problem switches from being in~\textbf{P} to being
\textbf{NP}-complete. In fact, the change-over happens as soon as not all
tokens are identical.

\begin{theorem}\label{short-th2}\quad\\*
  Restricted to the case that the token sets are~$(k-1,1)$ (i.e., there is
  one special token and all others are identical),
  \textsc{Shortest-Token-Moves-Sequence} is
  {\normalfont\textbf{NP}}-complete.
\end{theorem}

\noindent
It is possible to prove this using most of the ideas from the proof in
Papadimitriou et al.~\cite{PR} that
`\textsc{Graph-Motion-Planning-With-One-Robot}' is \textbf{NP}-complete.
Motion planning of robot(s) on graphs is very closely related to
transformations between token configurations on graphs. Except now there
are some special tokens, the `robots', that have to be moved from an
initial position to a specific final position, while all other tokens are
just `obstacles', and their final position is not relevant. The full
details of the proof of Theorem~\ref{short-th2} will appear in
Trakultraipruk~\cite{Tr}.

\subsection{Unlabelled Tokens with Restrictions}

If we consider the token problems in the previous subsection for the case
that all tokens are identical, then there is very little to prove. The
puzzle graph $\puz(G;k)$ (with $k\le|V(G)|$) is connected if and only if
$k=|V(G)|$ or~$G$ is connected. More specifically, two token configurations
are in the same component of $\puz(G;k)$ if and only if they have the same
number of tokens on each component of~$G$. Even finding the minimum number
of steps to go from one given configuration to another can be done in
polynomial time. (Of course, this does not mean that questions about other
properties of this kind of reconfiguration graphs cannot be interesting;
see for instance Fabila-Monroy et al.~\cite{FF}.)

But the situation changes drastically if only certain positions of tokens
are allowed. The following problem was studied in Hearn~\&
Demaine~\cite{HD}. Recall that a \emph{stable set} in a graph is a set of
vertices so that no two in the set are adjacent.

\bigskip
\begin{probl}{Stable-Sliding-Token-Configurations}
\item[\normalfont\textit{Instance}:\hfill] A graph~$G$, and two token
  configurations on~$G$ using identical tokens so that the set of occupied
  vertices for both configurations forms a stable set in~$G$.
\item[\normalfont\textit{Question}:\hfill] Is it possible to transform the
  first given configuration into the second one by a sequence of moves of
  one token along an edge, and such that in every intermediate
  configurations the set of occupied vertices is a stable set?
\end{probl}

\begin{theorem}[\normalfont Hearn~\&
  Demaine~\cite{HD}]\label{HD-th1}\quad\\*
  The problem \textsc{Stable-Sliding-Token-Configurations} is
  {\normalfont\textbf{PSPACE}}-complete,\linebreak
  even when restricted to planar graphs with maximum degree three.
\end{theorem}

\noindent
The proof of this theorem in~\cite{HD} (and many other results in that
paper) rely on a powerful general type of problem that seems to be very
suitable for complexity theoretical reductions. A \emph{non-deterministic
  constraint logic machine} (\emph{NCL machine}) consists of an undirected
graph, together with assignments of non-negative integer weights to its
edges and its vertices. A feasible configuration of an NCL machine is an
orientation of the edges such that the sum of incoming edge-weights at each
vertex is at least the weight of that vertex. A move is nothing other than
reversing the orientation of one edge, guaranteeing that the resulting
orientation is still a feasible configuration.

The following is a natural reconfiguration question for NCL machines.

\bigskip
\begin{probl}{NCL-Configuration-to-Edges}
\item[\normalfont\textit{Instance}:\hfill] An NCL machine, a feasible
  configuration on that machine, and a specific edge of the underlying
  graph.
\item[\normalfont\textit{Question}:\hfill] Is there a sequence of moves
  such that all intermediate configurations are feasible, and ending in a
  feasible configuration in which the specified edge has its orientation
  reversed?
\end{probl}

\begin{theorem}[\normalfont Hearn~\&
  Demaine~\cite{HD}]\label{HD-th2}\quad\\*
  The problem \textsc{NCL-Configuration-to-Edges} is
  {\normalfont\textbf{PSPACE}}-complete, even when restricted to NCL
  machines in which the underlying graph is planar, all vertices have
  degree three, all edge-weights are~1 or~2, and all vertex weights are~2.
\end{theorem}

\noindent
We return to moving tokens configuration problems. Note that in the
\textsc{Stable-Sliding-Token-Configurations} problem, the graph has a
`double' role: it determines both the allowed configurations (stable vertex
sets) and the allowed moves (sliding along an edge). A natural next
question would be what happens when one of these constraints imposed by the
graph is removed. We have already seen that if we remove the constraint
that the configurations must be stable sets, then the problem becomes easy.
But the situation is different if we remove the constraint that token
movement must happen along an edge.

\bigskip
\begin{probl}{Stable-Set-Reconfiguration}
\item[\normalfont\textit{Instance}:\hfill] A graph~$G$, and two token
  configurations on~$G$ using identical tokens so that the set of occupied
  vertices for both configurations forms a stable set in~$G$.
\item[\normalfont\textit{Question}:\hfill] Is it possible to transform the
  first given configuration into the second one by a sequence of moves of
  one token at each step, where a token can move from any vertex to any
  other vacant one, and so that in every intermediate configuration the set
  of occupied vertices is a stable set?
\end{probl}

\begin{theorem}[\normalfont Ito et al.~\cite{ID}]\label{ID-th1}\quad\\*
  The problem \textsc{Stable-Set-Reconfiguration} is
  {\normalfont\textbf{PSPACE}}-complete, even when restricted to planar
  graphs with maximum degree three.
\end{theorem}

\noindent
Since independent set problems are easily reduced to problems about
cliques, vertex covers, etc., reconfiguration problems where the vertices
occupied by a token form sets of this type are easily seen to be
\textbf{PSPACE}-complete as well. See Ito et al.~\cite{ID} for more
details.

For some other types of sets formed by occupied vertices, the corresponding
reconfiguration problems can become polynomial. A classical example of this
is the following.

\begin{theorem}[\normalfont Cummins~\cite{Cu}]\label{Cu-th1}\quad\\*
  Let~$G$ be a connected graph with positive weights on its edges. Then any
  minimum spanning tree of~$G$ can be transformed into any other minimum
  spanning tree by exchanging one edge at a time, so that each intermediate
  configuration is a minimum spanning tree as well.
\end{theorem}

\noindent
Note that the reconfiguration in Theorem~\ref{Cu-th1} can be seen as a
token reconfiguration problem by playing on the line graph of~$G$.
Similarly, the following problem is essentially
\textsc{Stable-Set-Reconfiguration} played on line graphs.

\bigskip
\begin{probl}{Matching-Reconfiguration}
\item[\normalfont\textit{Instance}:\hfill] A graph~$G$, and two matchings
  of~$G$ (subgraphs of degree at most one).
\item[\normalfont\textit{Question}:\hfill] Is it possible to transform the
  first matching into the second one by a sequence of moves of one edge at
  a time, so that in every intermediate configuration the set of chosen
  edges forms a matching as well?
\end{probl}

\begin{theorem}[\normalfont Ito et al.~\cite{ID}]\label{ID-th2}\quad\\*
  The problem \textsc{Matching-Reconfiguration} is in
  {\normalfont\textbf{P}}.
\end{theorem}

\noindent
Comparing the reconfiguration problems in this subsection that are
\textbf{PSPACE}-complete with those that are in~\textbf{P}, it is tempting
to conjecture that if the related \emph{decision problem} is
\textbf{NP}-complete, then the reconfiguration problem is
\textbf{PSPACE}-complete; whereas if the related decision problem is
in~\textbf{P}, then so is the reconfiguration problem. Such a connection is
alluded to in Ito et al.~\cite{ID}. Nevertheless, in earlier sections we
have seen some examples that shows that such a direct connection is not
true. For instance, it is \textbf{NP}-complete to decide if a graph is
3-colourable, but the single-vertex recolouring reconfiguration problem is
in~\textbf{P}, Theorem~\ref{col-th1}\,(b).

We close this section with a simplified version of a question from Ito et
al.~\cite{ID}: is the \textsc{Hamilton-Cycle-Reconfiguration} problem
(where two cycles are adjacent if they differ in two edges)
\textbf{PSPACE}-complete?

\section{Applications}\label{sec-appl}

Most reconfiguration problems are interesting enough for their own sake,
and do not really need an application to justify their study. Nevertheless,
many reconfiguration problems have applications or are inspired by problems
in related areas. In this section we look at some of those applications and
connections.

\subsection{Sampling and counting}\label{subs-sc}

Randomness plays an important role in many parts of combinatorics and
theoretical computer science. Indeed, results from probability theory have
led to major developments in both fields. It is therefore unsurprising that
researchers are often interested in obtaining random samples of particular
combinatorial structures. For example, much attention has been devoted to
the problem of sampling from an exponential number of structures
(exponential in the size of the object over which the structures are
defined) in time polynomial in this quantity. One of the reasons for this
is that being able to sample almost uniformly from a set of combinatorial
structures is enough to be able to approximately count such structures. See
Jerrum~\cite{Je1} for an example illustrating the method in the context of
graph colourings, and Jerrum~\cite{Je2} and Jerrum et al.~\cite{JVV} for
full details.

The question of when the configuration graph of a reconfiguration problem
is connected is quite old. In particular the configuration graph of the
single-vertex recolouring method has been looked at, as a subsidiary issue,
by researchers in the statistical physics community studying the `Glauber
dynamics of an anti-ferromagnetic Potts model at zero temperature'. (See
Sokal~\cite{So} for an introduction to the Potts model and its many
relations to graph theory.) Associated with that research is the work on
rapid mixing of Markov chains used to obtain efficient algorithms for
almost uniform sampling of $k$-colourings of a given graph. We give a brief
description of the basic ideas involved in these areas of research.

Quite often, the sampling is done via the simulation of an appropriately
defined Markov chain. Here the important point is that the Markov chain
should be rapidly mixing. This means, loosely speaking, that it should
converge to a close approximation of the stationary distribution in time
polynomial in the size of the problem instance. For a precise description
of this concept and further details; see~\cite{Je2} again.

In the context of the particular Markov chain used for sampling
$k$-colourings of a graph known as \emph{Glauber dynamics} (originally
defined for the \emph{anti-ferromagnetic Potts model at zero temperature})
we have the following. For a particular graph~$G$ and value of~$k$, let us
denote the Glauber dynamics for the $k$-colourings of~$G$ by
$\mathcal{M}_k(G)=(X_t)_{t=0}^\infty$. The state space of
$\mathcal{M}_k(G)$ is the set of all $k$-colourings of~$G$, the initial
state~$X_0$ is an arbitrary colouring, and its transition probabilities are
determined by the following procedure.

\begin{enumerate}
\item Select a vertex~$v$ of~$G$ uniformly at random.

\vspace{-2mm}\item Select a colour $c\in\{1,2,\ldots,k\}$ uniformly at random.

\vspace{-2mm}\item If recolouring vertex~$v$ with colour~$c$ yields a
proper colouring, then set $X_{t+1}$ to be this new colouring; otherwise,
set $X_{t+1}=X_t$.
\end{enumerate}

The relation between $\mathcal{M}_k(G)$ and the single-vertex recolouring
transformations is immediate. In particular, to be sure that every
$k$-colouring can appear as some state of the Markov chain, we need
that~$G$ is $k$-mixing. Thus the fact that a graph is $k$-mixing is a
necessary condition for its Glauber dynamics Markov chain to be rapidly
mixing. (This explains the choice of terminology in Section~\ref{sec-col}.)
On the other hand, if a graph is $k$-mixing it does not mean that its
Glauber dynamics Markov chain is rapidly mixing. An example showing this is
given by the stars $K_{1,m}$, which are $k$-mixing for any $k\ge3$ (see
Theorem~\ref{DF-th1}) but whose Glauber dynamics is not rapidly mixing for
$k\le m^{1-\varepsilon}$, for fixed $\varepsilon>0$ (\L uczak~\&
Vigoda~\cite{LV}).

Let us point out that much of the work on rapid mixing of the Glauber
dynamics Markov chain (as well as that on its many generalisations and
variants) has concentrated on specific graphs, or on values of~$k$ so large
that $k$-mixing is guaranteed. Many applications in theoretical physics
related to the Potts model are of particular interest for crystalline
structures, leading to the many studies of the Glauber dynamics and its
generalisations on very regular and highly symmetric graphs such as integer
grids.

Similar to the single-vertex recolouring method, the Kempe chain
recolouring method (see Subsection~\ref{subs-Kc}) has also been used to
define a Markov chain on the set of all $k$-colourings of some graph. The
corresponding approximate sampling algorithm is known as the
\emph{Wang-Swendsen-Koteck\'y dynamics}; see \cite{WSK1,WSK2}.

\medskip
Many reconfiguration problems we have considered so far can be used to
define a Markov chain similar to the ones for vertex-colouring defined
above. Again, for such a Markov chain to be a useful tool for almost
uniform sampling and approximate counting, it is necessary that the
configuration graph is connected, leading to questions considered in this
survey.

\subsection{Puzzles and games}\label{subs-gp}

We introduced the study of token configurations on graphs by looking at the
classical 15-puzzle. But in fact, many puzzles and games can be described
as reconfiguration problems. Following Demaine and Hearn~\cite{DH} we use
the term \emph{(combinatorial) puzzle} if there is only one player, and use
\emph{(combinatorial) game} if there are two players. (So we ignore games
with more than two players, or with no players (like Conway's Game of
Life).)

The puzzles we are interested in are of the following type: ``Given some
initial configuration and a collection of allowed moves, can some
prescribed final configuration (or a final configuration from a prescribed
set) be reached in a finite number of moves?'' For a game the situation is
somewhat more difficult, and several different variants can be described. A
quite general one is: ``Given some initial configuration, a starting
player, a collection of allowed moves which the two players have to play
alternately, and a collection of winning configurations for player~1, can
player~1 force the game to always reach a winning configuration in a finite
number of moves, no matter the moves player~2 chooses?'' Another way to
describe this question is: ``Given the setup of the game and the initial
situation, is there a winning strategy for player~1?''.

With these descriptions there is an obvious relation between the type of
reconfiguration problem we considered and puzzles and games. For many
puzzles and games, both existing and specially invented, the complexity of
answering the questions above have been considered. A good start to find
the relevant results and literature in this area is the extensive survey of
Demaine and Hearn~\cite{DH}.

\subsection{Other applications}\label{subs-oa}

Some reconfiguration problems have more practical applications (leaving
aside if ``solving puzzles'' is really a practical application). In
particular, graph recolouring problems can be seen as abstract versions of
several real-life problems. One example of this is as a modelling tool for
the assignment of frequencies in radio-communication networks. The basic
aim of the \emph{Frequency Assignment Problem (FAP)} is to assign
frequencies to users of a wireless network, minimising the interference
between them and taking care to use the smallest possible range of
frequencies. Because the radio spectrum is a naturally limited resource
with a constantly growing demand for the services that rely on it, it has
become increasingly important to use it as efficiently as possible. As a
result, and because of the inherent difficulty of the problem, the subject
is huge. For an introduction and survey of different approaches and results
we refer the reader to Leese~\& Hurley~\cite{LH}.

The FAP was first defined as a graph colouring problem by Hale~\cite{Ha}.
In this setting, we think of the available frequencies (discretised and
appropriately spaced in the spectrum) as colours, transmitters as vertices
of a graph, and we add edges between transmitters that must be assigned
different frequencies. In order to better capture the subtleties of the
`real-world' problem, this basic model has been generalised in a multitude
of different ways. Typically this might involve taking into account the
fact that radio waves decay with distance obeying an inverse-square law.
For instance, numerical weights can be placed on the edges of the graph to
indicate that frequencies assigned to the end-vertices of an edge must
differ by at least the amount given by the particular edge-weight.

One of the major factors contributing to the rise in demand for use of the
radio spectrum has been the dramatic growth, both in number and in size, of
mobile telecommunication systems. In such systems, where new transmitters
are continually added to meet increases in demand, an optimal or
near-optimal assignment of frequencies will in general not remain so for
long. On the other hand, it might just be the case that, because of the
difficulty of finding optimal assignments, a sub-optimal assignment is to
be replaced with a recently-found better one. It thus becomes necessary to
think of the assignment of frequencies as a dynamic process, where one
assignment is to be replaced with another. In order to avoid interruptions
to the running of the system, it is desirable to avoid a complete
re-setting of the frequencies used on the whole network. In a graph
colouring framework, this leads naturally to the graph-recolouring problems
we looked at in Section~\ref{sec-col}.

Not much attention seems yet to have been devoted to the problem of
reassigning frequencies in a network. Some first results can be found in
\cite{BJ,BL,Han,MH}. Most of the work in that literature describes specific
heuristic approaches to the problem, often accompanied by some
computational simulations.

\medskip
As hinted at already in Section~\ref{sec-tok}, moving-token puzzles are
related to questions on movements of robots. A simple abstraction and
discretisation is to assume that one or more robots move along the edges of
a graph. There might be additional objects placed on the vertices, playing
the role of `obstacles'. The robots have to move from an initial
configuration to some target configuration. In order to pass vertices
occupied by obstacles, these obstacles have to be moved out of the way,
along edges as well.

In general it is not hard to decide whether or not the robots can actually
move from their initial to their target configuration. But for practical
applications, limiting the number of steps required is also important,
leading to problems that are much harder to answer, see e.g.\ Papadimitriou
et al.~\cite{PR}. A multi-robot motion planning problem in which robots are
partitioned in groups such that robots in the same group are
interchangeable, comparable to the sliding token problem with different
types of tokens, has recently been studied in Solovey~\&
Halperin~\cite{SH}.

The robot motion problem on graphs is closely related to certain puzzles as
well. A well-known example of such a puzzle is \emph{Sokoban}, which is
played on a square grid where certain squares contain immovable walls or
other obstacles. There is also a single `pusher' who can move certain
blocks from one square to a neighbouring one, but only in the direction the
pusher can `push'. Moreover, the pusher can only move along unoccupied
squares. The goal of the game is to push the movable blocks to their
prescribed final position. Deciding if a given Sokoban configuration can be
solved is known to be \textbf{PSPACE}-complete, as was proved in
Culberson~\cite{Cul}.


\section{Open problem}

It would be easy to end this survey with a long list of open problems: what
is the complexity of deciding the following reconfiguration problems:
\ldots~? But a more fundamental, and probably more interesting, problem is
to try to find a connection between the complexity of reconfiguration
problems and the complexity of the decision problem on the existence of
configurations of a particular kind (related to the reconfiguration problem
under consideration).

Such connections are regularly alluded to in the literature. For instance,
with regard to the complexity of satisfiability reconfiguration problems,
Gopalan et al.~\cite{GK1} conjectured that if~$\mathcal{S}$ is Schaefer,
then \textsc{Conn($\mathcal{S}$)} is in~$\textbf{P}$. (This has been
disproved since then; see Section~\ref{sec-sat}.) Similarly, Ito et
al.~\cite{ID} write ``There is a wealth of reconfiguration versions of
NP-complete problems which can be shown PSPACE-complete via extensions,
often quite sophisticated, of the original NP-completeness proofs;
\ldots'', as if there is a general connection between the complexity of
these two types of decision problem.

But the connection must be more subtle that just ``\textbf{NP}-completeness
of decision problem implies \textbf{PSPACE}-completeness of corresponding
reconfiguration problem''. For instance, it is well known that deciding if
a graph is $k$-colourable is \textbf{NP}-complete for any fixed $k\ge3$.
But deciding if two given 3-colourings of a graph are connected via a
sequence of single-vertex recolourings is in~\textbf{P} for $k=3$ and
\textbf{PSPACE}-complete for $k\ge4$; see Theorem~\ref{col-th1}.

Nevertheless, it might be possible to say more about the connection between
the complexity of certain decision problems and the complexity of the
corresponding reconfiguration problem. In particular for problems that
involve labelling of certain objects under constraints (such as
satisfiability and graph colouring), and where the allowed transformation
is the relabelling of a single object, such a connection might be
identifiable. If this is indeed the case, it might give us a better
understanding of both the original decision problem and the reconfiguration
problem.

\thankyou{The author likes to thank the anonymous referee for very careful
  reading and for suggestions that greatly improved the presentation in the
  survey.}

\myaddress


\begin{thebibliography}{11}

\small
\bibitem{AH1}
  \writer{K. Appel and W. Haken}
  \paper{Every planar map is four colourable. I.~Discharging}
  \jnl{Illinois~J. Math.}
  \vol{21}
  \jyear{1977}
  \pages{429--490}

\bibitem{AH3}
  \writer{K. Appel and W. Haken}
  \paper{Every planar map is four colourable}
  \jnl{Contemporary Mathematics}
  \vol{98}
  \pyear{1989}
  \publish{American Mathematical Society}{Providence, RI.}

\bibitem{AH2}
  \writer{K. Appel, W. Haken, and J. Koch}
  \paper{Every planar map is four colourable. II.~Reducibility}
  \jnl{Illinois~J. Math.}
  \vol{21}
  \jyear{1977}
  \pages{491--567}

\bibitem{AM}
  \writer{V. Auletta, A. Monti, M. Parente, and P. Persiano}
  \paper{A linear time algorithm for the feasibility of pebble motion on
    trees}
  \jnl{Algorithmica}
  \vol{23}
  \jyear{1999}
  \pages{223-245}

\bibitem{BJ}
  \writer{V. Barb\'era and B. Jaumard}
  \paper{Design of an efficient channel block retuning}
  \jnl{Mobile Netw.\ Appl.}
  \vol{6}
  \jyear{2001}
  \pages{501--510}

\bibitem{BL}
  \writer{J. Billingham, R.A. Leese, and H. Rajaniemi}
  \paper{Frequency reassignment in cellular phone networks}
  \series{Smith Institute Study Group Report}
  \tome{available from
    \url{www.smithinst.ac.uk/Projects/ESGI53/ESGI53-Motorola/Report}}
  \byear{2005}

\bibitem{BM}
  \writer{J.A. Bondy and U.S.R. Murty}
  \book{Graph Theory}
  \publish{Springer}{New York}
  \byear{2008}

\bibitem{BC}
  \writer{P. Bonsma and L. Cereceda}
  \paper{Finding paths between graph colourings: PSPACE-completeness and
    superpolynomial distances}
  \jnl{Theoret.\ Comput.\ Sci.}
  \vol{410}
  \jyear{2009}
  \pages{5215--5226}

\bibitem{BHT}
  \writer{G. Brightwell, J. van den Heuvel, and S. Trakultraipruk}
  \unpubperiod{Connectedness of token graphs with labelled tokens. In
    preparation}

\bibitem{BH}
  \writer{J.K. Burton Jr.\ and C.L. Henley}
  \paper{A constrained Potts antiferromagnet model with an interface
    representation}
  \jnl{J.~Phys.~A}
  \vol{30}
  \jyear{1997}
  \pages{8385--8413}

\bibitem{Ce}
  \writer{L. Cereceda}
  \unpublish{Mixing Graph Colourings}{PhD Thesis, London School of
    Economics}
  \byear{2007}

\bibitem{CHJ1}
  \writer{L. Cereceda, J. van den Heuvel, and M. Johnson}
  \paper{Connectedness of the graph of vertex-colourings}
  \jnl{Discrete Math.}
  \vol{308}
  \jyear{2008}
  \pages{913--919}

\bibitem{CHJ2}
  \writer{L. Cereceda, J. van den Heuvel, and M. Johnson}
  \paper{Mixing 3-colourings in bipartite graphs}
  \jnl{European J. Combin.}
  \vol{30}
  \jyear{2009}
  \pages{1593--1606}

\bibitem{CHJ3}
  \writer{L. Cereceda, J. van den Heuvel, and M. Johnson}
  \paper{Finding paths between 3-colorings}
  \jnl{J.~Graph Theory}
  \vol{67}
  \jyear{2011}
  \pages{69--82}

\bibitem{CM}
  \writer{K. Choo and G. MacGillivray}
  \paper{Gray code numbers for graphs}
  \jnl{Ars Math.\ Contemp.}
  \vol{4}
  \jyear{2011}
  \pages{125--139}

\bibitem{Cl}
  \writer{Clay Mathematical Institute}
  \unpubperiod{The Millennium Prize Problems.
    \url{www.claymath.org/millennium/}}

\bibitem{CE}
  \writer{C.R. Cook and A.B. Evans}
  \paper{Graph folding}
  \proc{Proceedings of the 10th Southeastern Conference on Combinatorics,
    Graph Theory and Computing},
  \jnl{Congress.\ Numer.}
  \vol{XXIII-XXIV}
  \jyear{1979}
  \ppages{305--314}

\bibitem{Cul}
  \writer{J.C. Culberson}
  \paper{Sokoban is PSPACE-complete}
  \unpublish{Technical Report \textbf{TR 97-02}}{Department of Computing
    Science, University of Alberta}
  \tome{available via
    \url{citeseerx.ist.psu.edu/viewdoc/summary?doi=10.1.1.52.41}}
  \byear{1997}

\bibitem{Cu}
  \writer{R.L. Cummins}
  \paper{Hamilton circuits in tree graphs}
  \jnl{IEEE Trans.\ Circuit Theory}
  \vol{CT-13}
  \jyear{1966}
  \pages{82--90}

\bibitem{DH}
  \writer{E.D. Demaine and R.A. Hearn}
  \paper{Playing games with algorithms: Algorithmic combinatorial game
    theory}
  \tome{\texttt{arXiv:cs/0106019v2 [cs.CC]}}
  \byear{2008}

\bibitem{Di}
  \writer{R. Diestel}
  \book{Graph Theory}
  \publish{Springer}{Heidelberg}
  \byear{2010}

\bibitem{DF}
  \writer{M. Dyer, A. Flaxman, A. Frieze, and E. Vigoda}
  \paper{Randomly colouring sparse random graphs with fewer colours than
    the maximum degree}
  \jnl{Random Structures Algorithms}
  \vol{29}
  \jyear{2006}
  \pages{450--465}

\bibitem{FF}
  \writer{R. Fabila-Monroy, D. Flores-Pe\~naloza, C. Huemer, F. Hurtado, J.
    Urrutia, and D.R.~Wood}
  \paper{Token graphs}
  \jnl{Graphs Combin.}
  \vol{28}
  \jyear{2012}
  \pages{365--380}

\bibitem{FS}
  \writer{S.J. Ferreira and A.D. Sokal}
  \paper{Antiferromagnetic Potts models on the square lattice: A
    high-precision Monte Carlo study}
  \jnl{J.~Statist.\ Phys.}
  \vol{96}
  \jyear{1999}
  \pages{461--530}

\bibitem{GJ}
  \writer{M.R. Garey and D.S. Johnson}
  \book{Computers and Intractability: A Guide to the Theory of
    NP-completeness}
  \publish{Freeman}{New York}
  \byear{1979}

\bibitem{Go}
  \writer{O. Goldreich}
  \paper{Finding the shortest move-sequence in the graph-generalized
    15-puzzle is NP-hard}
  \proc{Studies in Complexity and Cryptography; Miscellanea on the
    Interplay between Randomness and Computation},
  \series{Lect.\ Notes Comput.\ Sci.}
  \vol{6650}
  \pyear{2011}
  \ppages{1--5}

\bibitem{GK1}
  \writer{P. Gopalan, P.G. Kolaitis, E. Maneva, and C.H. Papadimitriou}
  \paper{The connectivity of Boolean satisfiability: Computational and
    structural dichotomies}
  \proc{Proceedings of Automata, Languages and Programming, 33rd
    International Colloquium},
  \series{Lect.\ Notes Comput.\ Sci.}
  \vol{4051}
  \pyear{2006}
  \ppages{346--357}

\bibitem{GK2}
  \writer{P. Gopalan, P.G. Kolaitis, E. Maneva, and C.H. Papadimitriou}
  \paper{The connectivity of Boolean satisfiability: Computational and
    structural dichotomies}
  \jnl{SIAM J. Comput.}
  \vol{38} 
  \jyear{2009}
  \pages{2330--2355}

\bibitem{GH}
  \writer{G. Goraly and R. Hassin}
  \paper{Multi-color pebble motion on graphs}
  \jnl{Algorithmica}
  \vol{58}
  \jyear{2010}
  \pages{610-636}

\bibitem{Ha}
  \writer{W.K. Hale}
  \paper{Frequency assignment: Theory and applications}
  \jnl{Proc.\ IEEE}
  \vol{68}
  \jyear{1980}
  \pages{1497--1514}

\bibitem{Han}
  \writer{J. Han}
  \paper{Frequency reassignment problem in mobile communication networks}
  \jnl{Comput.\ Oper.\ Res.}
  \vol{34}
  \jyear{2007}
  \pages{2939--2948}

\bibitem{HD}
  \writer{R.A. Hearn and E.D. Demaine}
  \paper{PSPACE-completeness of sliding-block puzzles and other problems
    through the nondeterministic constraint logic model of computation}
  \jnl{Theoret.\ Comput.\ Sci.}
  \vol{343}
  \jyear{2005}
  \pages{72--96}

\bibitem{ID}
  \writer{T. Ito, E.D. Demaine, N.J.A. Harvey, C.H. Papadimitriou, M.
    Sideri, R. Uehara, and Y. Uno}
  \paper{On the complexity of reconfiguration problems}
  \jnl{Theoret.\ Comput.\ Sci.}
  \vol{412}
  \jyear{2011}
  \pages{1054--1065}

\bibitem{Je1}
  \writer{M. Jerrum}
  \paper{A very simple algorithm for estimating the number of
    $k$-colourings of a low degree graph}
  \jnl{Random Structures Algorithms}
  \vol{7}
  \jyear{1995}
  \pages{157--165}

\bibitem{Je2}
  \writer{M. Jerrum}
  \book{Counting, Sampling and Integrating: Algorithms and Complexity}
  \publish{Birk\-h\"auser Verlag}{Basel}
  \byear{2003}

\bibitem{JVV}
  \writer{M.R. Jerrum, L.G. Valiant, and V.V. Vazirani}
  \paper{Random generation of combinatorial structures from a uniform
    distribution}
  \jnl{Theoret.\ Comput.\ Sci.}
  \vol{43}
  \jyear{1986}
  \pages{169--188}

\bibitem{KM}
  \writer{D. Kornhauser, G. Miller, and P. Spirakis}
  \paper{Coordinating pebble motion on graphs, the diameter of permutation
    groups, and applications}
  \proc{Proceedings of the 25th Annual Symposium on Foundations of Computer
    Science}
  \pyear{1984}
  \ppages{241-250}

\bibitem{LM}
  \writer{M. Las Vergnas and H. Meyniel}
  \paper{Kempe classes and the Hadwiger Conjecture}
  \jnl{J.~Combin.\ Theory Ser.~B}
  \vol{31}
  \jyear{1981}
  \pages{95--104}

\bibitem{LH}
  \writer{R.A. Leese and S. Hurley (eds.)}
  \book{Methods and Algorithms for Radio Channel Assignment}
  \publish{Oxford Univ.\ Press}{Oxford}
  \byear{2003}

\bibitem{LV}
  \writer{T. \L uczak and E. Vigoda}
  \paper{Torpid mixing of the Wang-Swendsen-Koteck\'y algorithm for
    sampling colorings}
  \jnl{J.~Discrete Alg.}
  \vol{3}
  \jyear{2005}
  \pages{92--100}

\bibitem{MT}
  \writer{K. Makino, S. Tamaki, and M. Yamamoto}
  \paper{On the Boolean connectivity problem for Horn relations}
  \jnl{Discrete Appl.\ Math.}
  \vol{158}
  \jyear{2010}
  \pages{2024--2030}

\bibitem{MH}
  \writer{O. Marcotte and P. Hansen}
  \paper{The height and length of colour switching}
  \proc{Graph Colouring and Applications}
  \edits{P.~Hansen and O. Marcotte}
  \publish{AMS}{Providence}
  \pyear{1999}
  \ppages{101--110}

\bibitem{Me}
  \writer{H. Meyniel}
  \paper{Les 5-colorations d'un graphe planaire forment une classe de
    commutation unique}
  \jnl{J.~Combin.\ Theory Ser.~B}
  \vol{24}
  \jyear{1978}
  \pages{251--257}

\bibitem{Mo}
  \writer{B. Mohar}
  \paper{Kempe equivalence of colorings}
  \proc{Graph theory in Paris}
  \edits{J.A. Bondy, J.~Fonlupt, J.-L. Fouquet, J.-C. Fournier, and J.L.
    Ram{\'\i}rez Alfons{\'\i}n}
  \publish{Birkh\"auser Verlag}{Basel}
  \pyear{2007}
  \pages{287--297}

\bibitem{Pa}
  \writer{C.H. Papadimitriou}
  \book{Computational Complexity}
  \publish{Addison-Wesley}{Boston}
  \byear{1994}

\bibitem{PR}
  \writer{C.H. Papadimitriou, P. Raghavan, M. Sudan, and H. Tamaki}
  \paper{Motion planning on a graph}
  \proc{Proceedings of the 35th Annual Symposium on Foundations of Computer
    Science}
  \pyear{1994}
  \ppages{511--520}
  {Long version available online at:
    \url{people.csail.mit.edu/madhu/papers/1994/robot-full.pdf}.}

\bibitem{Sa}
  \writer{W.J. Savitch}
  \paper{Relationships between nondeterministic and deterministic tape
    complexities}
  \jnl{J.~Comput.\ System Sci.}
  \vol{4}
  \jyear{1970}
  \pages{177--192}

\bibitem{Sc}
  \writer{T. Schaefer}
  \paper{The complexity of satisfiability problems}
  \proc{Proceedings of the 10th Annual ACM Symposium on Theory of
    Computing}
  \pyear{1978}
  \ppages{216--226}

\bibitem{Sch}
  \writer{A. Schrijver}
  \book{Combinatorial Optimization; Polyhedra and Efficiency}
  \publish{Springer-Verlag}{Berlin}
  \byear{2003}

\bibitem{So}
  \writer{A.D. Sokal}
  \paper{The multivariate Tutte polynomial (alias Potts model) for graphs
    and matroids}
  \proc{Surveys in Combinatorics 2005}
  \publish{Cambridge Univ.\ Press}{Cambridge}
  \pyear{2005}
  \ppages{173--226}

\bibitem{SH}
  \writer{K. Solovey and D. Halperin}
  \paper{$k$-Color multi-robot motion planning}
  \tome{\texttt{arXiv:1202.6174v2 [cs.RO]}}
  \byear{2012}

\bibitem{Tr}
  \writer{S. Trakultraipruk}
  \unpublish{Connectivity Properties of Some Transformation Graphs}{PhD
    Thesis, London School of Economics}
  \tome{in preparation}
  \byear{2013}

\bibitem{Vi}
  \writer{V.G. Vizing}
  \paper{On an estimate of the chromatic class of a $p$-graph (in Russian)}
  \jnl{Metody Diskret.\ Analiz.}
  \vol{3}
  \jyear{1964}
  \pages{25--30}

\bibitem{WSK1}
  \writer{J.S. Wang, R.H. Swendsen, R. Koteck\'y}
  \paper{Antiferromagnetic Potts models}
  \jnl{Phys.\ Rev.\ Lett.}
  \vol{63}
  \jyear{1989}
  \pages{109--112}

\bibitem{WSK2}
  \writer{J.S. Wang, R.H. Swendsen, R. Koteck\'y}
  \paper{Three-state antiferromagnetic Potts models: A Monte Carlo study}
  \jnl{Phys.\ Rev.~B}
  \vol{42}
  \jyear{1990}
  \pages{2465--2474}

\bibitem{Wi}
  \writer{R.M. Wilson}
  \paper{Graph puzzles, homotopy, and the alternating group}
  \jnl{J.~Combin.\ Theory Ser.~B}
  \vol{16}
  \jyear{1974}
  \pages{86-96}

\end{thebibliography}
\end{document}